\documentclass[twocolumn,showpacs,preprintnumbers,final,a4paper,aps,citeautoscript,footinbib,pra,superscriptaddress]{revtex4-2}
\usepackage{amsmath,amssymb,amsfonts}
\usepackage{physics}
\usepackage{graphicx,color}% Include figure files
\graphicspath{{./figures/}}
\usepackage{float}
\usepackage{bm}% bold math
\usepackage{times}
\usepackage[colorlinks,bookmarks=true,citecolor=blue,linkcolor=blue,urlcolor=blue,breaklinks=true]{hyperref}
\usepackage{amsthm}
\newtheorem{thm}{Theorem}
\newtheorem{lem}[thm]{Lemma}

\newcommand{\hcc}{\hc{c}}
\newcommand{\hcd}{\hc{d}}

\newcommand{\tz}{\tau^{z}}
\newcommand{\rmm}[1]{{\rm{#1}}}

\newcommand{\intd}{\int \rmm{d}}

\newcommand{\del}[2]{ \frac{\partial #1}{\partial #2}}

\newcommand{\eno}[1]{\rmm{e}^{#1}}

\newcommand{\lr}[3] { \left#1 #2 \right#3}
\newcommand{\hc}[1]{ {#1} ^{\dagger} }

\newcommand{\naraba}{\Rightarrow}

\newcommand{\com}[2]{\lr{[}{#1,#2}{]} }
\newcommand{\kitaiti}[1]{ \lr{\langle} {#1} {\rangle} }

\newcommand{\chuukakko}[1]{\lr{\{}{#1}{\}}}

\newcommand{\acom}[2]{\chuukakko{#1,#2}}

\newcommand{\migi}{\rightarrow}

\newcommand{\hf}{\frac{1}{2}}

\newcommand{\eps}{\epsilon}

\newcommand{\sx}{\sigma^{x}}
\newcommand{\sy}{\sigma^{y}}
\newcommand{\sz}{\sigma^{z}}

\newcommand{\im}{\rmm{i}}
\newcommand{\nt}{\notag\\}
\newcommand{\HC}{+\rmm{H.c.}}

\begin{document}
%%%%%%%%%%%%%%%%%%%%%%%%%%%%%%%%%%%%%%%%%%%%%%%%%%%%%%%%%%%%%%%
%%%%%%%%%%%%%%%%%%%%%%%%%%%%%%%%%%%%%%%%%%%%%%%%%%%%%%%%%%%%%%%
%\title{Supersymmetry Breaking and Superfrustration in a Cooper-pair-like Generalized Nicolai Model}
\title{Supersymmetry Breaking in a Generalized Nicolai Model with Fermion Pairing}
%%%%%%%%%%%%%%%%%%%%%%%%%%%%%%%%%%%%%%%%%%%%%%%%%%%%%%%%%%%%%%%
\author{Urei Miura} 
\email{urei.miura@yukawa.kyoto-u.ac.jp}
\affiliation{%
Division of Physics and Astronomy, Graduate School of Science, 
Kyoto University, Kyoto 606-8502, Japan}%
\affiliation{%
Center for Gravitational Physics and Quantum Information, 
Yukawa Institute for Theoretical Physics, Kyoto University, Kitashirakawa Oiwake-Cho, Kyoto 606-8502, Japan
}%
%%%%%%%%%%%%%%%%%%%%%%%%%%%%%%%%%%%%%%%%%%%%%%%%%%%%%%%%%%
\author{Keisuke Totsuka}
\affiliation{%
Center for Gravitational Physics and Quantum Information, 
Yukawa Institute for Theoretical Physics, Kyoto University, Kitashirakawa Oiwake-Cho, Kyoto 606-8502, Japan
}%
%%%%%%%%%%%%%%%%%%%%%%%%%%%%%%%%%%%%%%%%%%%%%%%%%%%%%%%%%%
\date{\today}
%%%%%%%%%%%%%%%%%%%%%%%%%%%%%%%%%%%%%%%%%%%%%%%%%%%%%%%%%%
\begin{abstract}
We introduce a supersymmetric lattice fermion model that contains both fermion pairing and the interacting Nicolai model.  
This model possesses a single control parameter, $g$, introduced through the anticommutator of the supersymmetry generators 
(supercharge), and it is shown that supersymmetry is broken in both finite and infinite systems as long as $g$ is finite.  
%This mathematical statement is physically interpreted by the behaviors of the superconducting and supersymmetry-breaking order parameters.  
Additionally, the single-mode approximation is employed to establish an upper bound on the dispersion relation of Nambu-Goldstone (NG) fermions, demonstrating their gaplessness.  
Finally, it is shown, through numerical and analytic calculations, that the anticipated extensive ground-state degeneracy and 
a zero-energy flat band does not occur for generic values of the parameter $g$.  
%%
%We conclude by suggesting future research that could investigate the physical implications of supersymmetry breaking, the behavior of Nambu-Goldstone fermions, and the effects of fluctuations beyond mean-field theory on the behavior of the system near the supersymmetry breaking transition.
\end{abstract}
%%%%%%%%%%%%%%%%%%%%%%%%%%%%%%%%%%%%%%%%%%%%%%%%%%%%%%%%%
\maketitle
%%%%%%%%%%%%%%%%%%%%%%%%%%%%%%%%%%%%%%%%%%%%%%%%%%%%%%%%%
%\tableofcontents

%%%%%%%%%%%%%%%%%%%%%%%%%%%%%%%%%%%%%%%%%%%%%%%%%%%%%%%%%
\section{INTRODUCTION}
%%%%%%%%%%%%%%%%%%%%%%%%%%%%%%%%%%%%%%%%%%%%%%%%%%%%%%%%%
Supersymmetry (SUSY) is a fermionic symmetry defined by the super-algebra of $Q$ (supercharge) and $H$ (Hamiltonian) as given in the following equation \cite{witten1981dynamical,witten1982constraints,cooper1995supersymmetry}
\begin{align}
\label{superalgebra}
    Q^2=\qty(\hc{Q})^2=0,\ \acom{Q}{\hc{Q}}=H.
\end{align}
Supersymmetry has been a topic of intensive research in theoretical physics due to its potential to unify the fundamental forces in nature.  
This is accomplished by introducing a new set of particles that are supersymmetric partners of those existing in the Standard Model of particle physics. 
The existence of these new particles is predicted by various theoretical models, although they have yet to be discovered experimentally.
Supersymmetry has been applied to random systems in statistical physics
\cite{efetov1999supersymmetry,cecotti1983stochastic,gozzi1993stochastic,parisi1979random}, and has been discussed also in lattice models \cite{nicolai1976supersymmetry,nicolai1977extensions,fendley2003lattice,fendley2005exact,huijse2008superfrustration,huijse2012supersymmetric,sannomiya2016supersymmetry,sannomiya2017supersymmetry,sannomiya2019supersymmetry,rahmani2015emergent}, cold-atom systems
\cite{yu2008supersymmetry,yu2010simulating,snoek2005ultracold,snoek2006theory,lozano20071+,shi2010supersymmetric,lai2015relaxation,blaizot2015spectral,blaizot2017goldstino,tajima2021goldstino}, and topological insulators 
\cite{grover2014emergent} in condensed matter physics.

Supersymmetry has significant implications for our understanding of the universe at both the smallest and largest length scales;  
at the smallest scale, it might help to solve the hierarchy problem by addressing the large discrepancy between the strength 
of gravity and the other fundamental forces of nature \cite{PhysRevD.13.974,PhysRevD.14.1667}.  
At the largest scale, on the other hand, it can provide a candidate for dark matter, which is believed to make up a substantial fraction of the matter in the universe.

Recently, supersymmetric lattice models \cite{fendley2003lattice,PhysRevLett.90.120402} have gained attention 
due to their possibility to exhibit extensive ground-state degeneracy called superfrustration \cite{PhysRevLett.95.046403,PhysRevLett.101.146406,huijse2008superfrustration,huijse2012supersymmetric} and supersymmetry breaking \cite{sannomiya2016supersymmetry,sannomiya2017supersymmetry,sannomiya2019supersymmetry}.  
In this paper, we focus on a supersymmetric model of interacting spinless fermions on a lattice that include a pairing term favoring superconducting order 
as well as various interactions.  As a control parameter is varied, the model interpolates between a trivial model in which 
supersymmetry is broken and the Nicolai model \cite{nicolai1976supersymmetry,nicolai1977extensions} where it remains unbroken.   
Therefore, we expect a SUSY-restoring phase transition to occur as a consequence of the competition among several parts of 
the Hamiltonian.    

The primary objective of this paper is to investigate the ground-state and low-energy properties of a lattice model 
of interacting spinless fermions that contain paring terms with both analytical and numerical methods. 
Our key findings include a rigorous proof, based on inequalities, of broken supersymmetry for any finite value of the parameter. 
Using the single-mode approximation, we demonstrate that supersymmetry breaking leads to the emergence 
of Nambu-Goldstone (NG) fermions (Goldstinos) with gapless linear dispersion. 

In our model, interactions appear only in the second order in the control parameter, and this suggests extensive degeneracy in the ground 
state and a zero-energy flat fermion band for sufficiently small values of the parameter.   We also investigate the fate of 
this extensive ground-state degeneracy associated with the zero-energy flat band to answer the question negatively.  

The organization of this paper is as follows. In Sec.~\ref{sec:Model}, we first introduce a supersymmetric lattice model  
of interacting spinless fermions that contain a single control parameter ($g$) 
through the anticommutator of the supercharges, and then describe various symmetries and conserved quantities of the model.  
In Sec.~\ref{sec:SUSYSSB}, it is shown that supersymmetry is broken both in finite and infinite systems as long as $g$ is finite.  
This is confirmed by numerical calculations of order parameters.  
Additionally, a variational approach is used to establish an upper bound on the dispersion relation of Nambu-Goldstone (NG) fermions, 
proving their gaplessness.  
In Sec.~\ref{sec:SF}, the possibility of a flat band, which is expected to arise for small enough $g$ due to the absence of $d$-fermions in $H_\rmm{pair}$, 
is discussed.
We first numerically investigate how the spectrum evolves as the parameter $g$ is increased, which suggests that 
the zero-energy flat band disappears for finite $g$. 
Then, it is demonstrated analytically that 
although the band remains flat at the lowest order in $g$, it acquires a finite dispersion in the second-order perturbation.

%%%%%%%%%%%%%%%%%%%%%%%%%%%%%%%%%%%%%%%%%%%%%%%%%%%%%%%%%%
\section{MODEL}
\label{sec:Model}
%%%%%%%%%%%%%%%%%%%%%%%%%%%%%%%%%%%%%%%%%%%%%%%%%%%%%%%%%%
In this section, we introduce a supersymmetric lattice model and describe its symmetry.  
An anticommuting (fermionic) operator $Q$ called supercharge plays a primary role in supersymmetric quantum systems; 
once we identify the operator $Q$, the Hamiltonian is defined by: 
$H =\acom{Q }{\hc{Q} },$ and all the non-trivial spectral properties follow immediately from 
the following fundamental relations
\[  
\begin{split}
& [ H , Q ] = [ H , \hc{Q} ]=0 \; , \\
& \{ (-1)^{F} , Q \}= \{ (-1)^{F} , \hc{Q} \} =0  
\end{split} 
\]   
with $(-1)^{F}$ being an operator called the fermion parity.  

One typical way of constructing supersymmetric lattice models is to introduce lattice fermions $f_i$ and bosons $b_i$ 
satisfying $\comm{f_i}{b_i}=\comm{f_i}{\hc{b}_i}=0$ and the standard (anti)commutation relations, 
and then define the supercharge as \cite{nicolai1976supersymmetry}
\begin{equation}
Q:= \sum_{i} \hc{b}_i f_i \; .
\label{eqn:Q-generic}
\end{equation}  
For instance, usual bosons are commonly used for $b_{i}$ to realize supersymmetry in Bose-Fermi mixtures 
\cite{yu2008supersymmetry,yu2010simulating,snoek2005ultracold,snoek2006theory,lozano20071+,%
shi2010supersymmetric,lai2015relaxation,blaizot2015spectral,blaizot2017goldstino,tajima2021goldstino}.  
If we replace the boson $b_{i}$ with other {\em bosonic} operators, we can devise a variety of supersymmetric models.

%%%%%%%%%%%%%%%%%%%%%%%%%%%%%%%%%%%%%%%%%%%%%%%%%%%%%%%
\subsection{Supercharge and Hamiltonian}
%%%%%%%%%%%%%%%%%%%%%%%%%%%%%%%%%%%%%%%%%%%%%%%%%%%%%%%
In this paper, we consider a system of spinless fermions on a one-dimensional lattice of length $2 L$.   
For reasons of clarity, we consider the zigzag geometry shown in Fig.~\ref{def-Q}(a) and divide the entire lattice into two sublattices.   
On these sublattices, we respectively define the spinless fermions $c_j$ and $d_j$ which obey the standard canonical anticommutation relations:
\begin{align}
& \acom{c_{i}}{\hcc_{j}}=\acom{d_{i}}{\hcd_{j}}=\delta_{i,j} \nt
&\acom{c_{i}}{c_{j}}=\acom{d_{i}}{d_{j}}=\acom{c_{i}}{d_{j}}=\acom{c_{i}}{\hcd_{j}}=0
\end{align}
for all $i,j=1,\ldots,L$. We denote the number operators for $c_j$ and $d_j$ by $n^{c}_j=c_j^\dagger c_j$ and $n^{d}_j=d_j^\dagger d_j$ 
( $j=1,\ldots,L$), respectively.  
 
In this paper, we take a fermion-pairing term $c^{\dagger}_{i}c^{\dagger}_{i+1}$ as the bosonic operator that is to be paired with the fermion $d_{i}$ 
and consider a model characterized by the following supercharge [see Fig.~\ref{def-Q}(a)]:
\begin{equation}
Q_{\rmm{S}} \qty(g) : = 
\sum_{i=1}^L \qty(1+g\, \hc{c}_{i}\hc{c}_{i+1})d_i   \; .
\label{eqn:def-super-charge}
\end{equation} 
The above supercharge $Q_{\rmm{S}} \qty(g)$ has a physical interpretation that it creates a real-space Cooper pair 
while annihilating a $d_i$ fermion nearby. 
The nilpotency $Q_{\rmm{S}} \qty(g)^2=0$ follows immediately from $d_{i}^{2}=0$. 
As we can change the sign of $g$ by a gauge transformation: $c_i\migi \eno{\im \frac{\pi}{2}}c_i$, we can set $g\geq0$ without loss of generality. 
In the limit $g\rightarrow \infty$, $Q_{\rmm{S}} \qty(g)/g$ reduces to the supercharge of the Nicolai model 
\cite{nicolai1976supersymmetry}  
\[
Q_{\rmm{Nic}} = \sum_{j=1}^L  c_{2j-1} c^{\dagger}_{2j} c_{2j+1}  
\]
after the redefinition of the fermion operators: $\hc{c}_{j} \to c_{2j - 1}$,  $d_{j} \to - \hc{c}_{2j}$.    

As has been mentioned above, the Hamiltonian is given by the anticommutator of the supercharges:
\begin{equation}
 H_{\rmm{S}} \qty(g) =\acom{Q_{\rmm{S}} \qty(g)}{\hc{Q}_{\rmm{S}} \qty(g)}=L+g H_{\rmm{pair}}+g^{2} H_{\rmm{Nic}}
 \label{BCS-like}
\end{equation}
with
\begin{subequations}
\begin{align}
H_{\rmm{pair}} =& \sum_{i=1}^L\qty(
\hc{c}_{i}\hc{c}_{i+1}+c_{i+1} c_i ) \; ,  \label{H-pair}  \\
\begin{split}
H_{\rmm{Nic}} = & \sum_{i=1}^L \qty(\hc{c}_{i}c_{i+2}\hc{d}_{i+1}d_{i}+\hc{c}_{i+2}c_{i}\hc{d}_{i}d_{i+1}) \\
& - \sum_{i=1}^L \qty(n^c_i+n^c_{i+1}-1)n^d_i + \sum_{i=1}^L n^c_i n^c_{i+1} \; .   
\end{split}
\label{H-nic} 
\end{align}
\end{subequations}
The fermion parity operator that anticommutes with $Q_{\rmm{S}}(g)$ and $\hc{Q}_{\rmm{S}} \qty(g)$ reads:
\begin{equation}
(-1)^{F} := (-1)^{\sum_{i}(n^c_i + n^{d}_i)} \; .
\end{equation}
We can readily check that it commutes with $H_{\rmm{S}} \qty(g)$ as is expected from the general construction 
of SUSY Hamiltonians.  

The bilinear part $H_{\text{pair}}$ may be viewed as a $p$-wave pairing term
of a one-dimensional spinless lattice superconductor.  
Like the BCS Hamiltonian, it explicitly breaks the charge U(1) symmetry, and as will be shown numerically in Sec.~\ref{sec:SUSYSSB}, 
the gap function $\Delta:=\kitaiti{c_ic_{i+1}}$ (with $\kitaiti{\cdots}$ standing for the ground-state expectation value) takes a non-zero value for a certain range of a coupling constant $g$. 
A similar term is known to appear in models such as the Kitaev chain 
\cite{kitaev2001unpaired} (or its spin counterpart \cite{lieb1961two}), 
where interesting phenomena occur by broken particle-number conservation.  

All the fermion interactions are contained in the Nicolai model $H_{\text{Nic}}$ \cite{nicolai1976supersymmetry};   
the first term of $H_{\text{Nic}}$ represents pair hopping of fermions separated by at most three sites [see the arrows in Fig.~\ref{def-Q}(a)], 
while the second and third terms respectively describe a repulsive interaction among $c_i$ fermions
and an attractive one between $c_i$ and $d_i$.

A remark is in order about the relation to other supersymmetric models 
studied recently \cite{sannomiya2016supersymmetry,sannomiya2017supersymmetry}.   
All these models and ours fall into a class of models that are characterized by supercharges of the form \eqref{eqn:Q-generic} 
with $\hc{b}_i f_i $ defined on three-site clusters.    

As has been mentioned above, we have chosen $f_i=d_i$ and $b^{\dagger}_{i} = 1+g\, c^{\dagger}_{i}c^{\dagger}_{i+1}$ 
in the model \eqref{BCS-like} considered here.  
On the other hand, $f_i=c_{2i-1}$ and a bond-density-wave operator $b^{\dagger}_i = g + c_{2i}^{\dagger} c_{2i+1}$ 
($f_i=c_{i}$ and $b_{i}^{\dagger}= g - c_{i-1}c_{i+1}$) are used in the extended Nicolai model \cite{sannomiya2016supersymmetry} 
(the $\mathbb{Z}_2$ Nicolai model \cite{sannomiya2017supersymmetry}).    
The three different constructions are illustrated in Figs.~\ref{def-Q}(a)-(c); the three-site clusters (i.e., the local units $b_{i}^{\dagger}f_{i}$ 
of the supercharges) and the bonds on which the bosonic operators $b_{i}^{\dagger}$ are defined are highlighted in 
cyan and red, respectively.   
In the model \eqref{BCS-like} and the extended Nicolai model [Fig.~\ref{def-Q}(b)], the three-site clusters are defined only on the upward triangles,
and hence the resulting models have two-sublattice structures.    
In the $\mathbb{Z}_2$-Nicolai model [Fig.~~\ref{def-Q}(c)], on the other hand, the local units live on 
both upward and downward triangles, making the model translationally invariant.

%%%%%%%%%%%%%%%%%%% FIG %%%%%%%%%%%%%%%%%%%%%%%%%%%%%%%%%%%%
\begin{figure}[htbp]
\begin{center}
    \includegraphics[width=\columnwidth,clip]{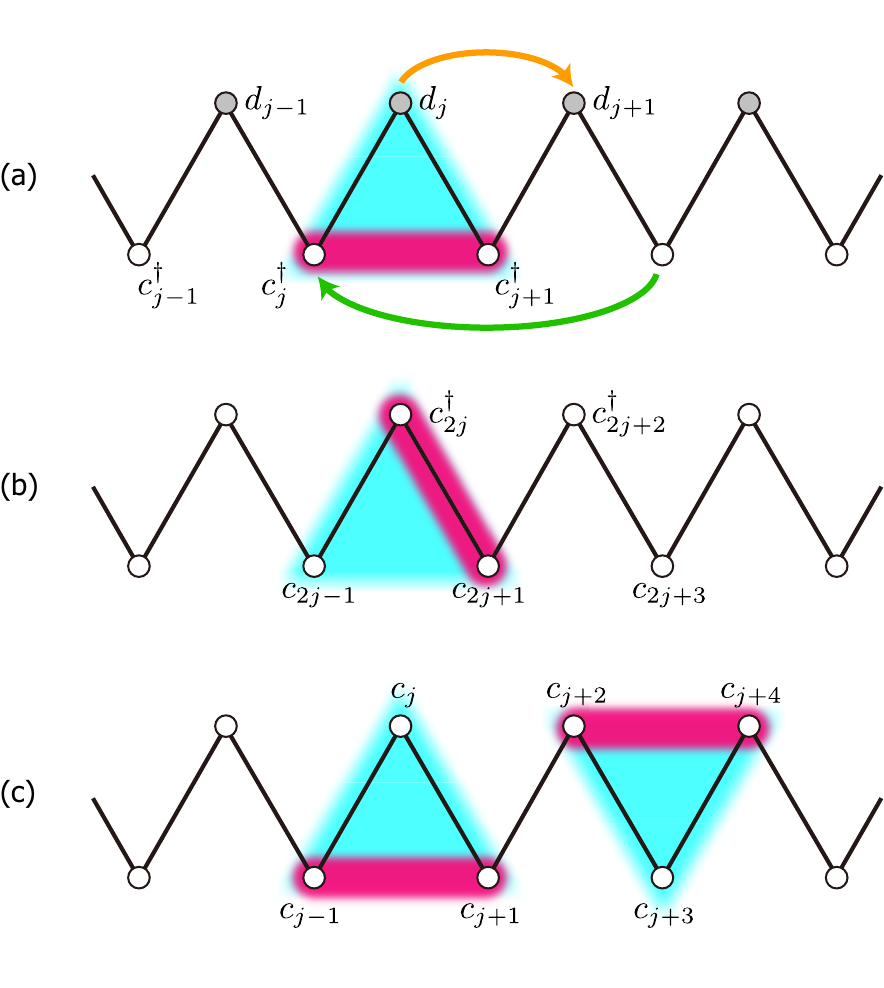}
\end{center}
\caption{Comparison of different supercharge constructions in known lattice fermion models: (a) the model \eqref{BCS-like} considered here, 
(b) the extended Nicolai model in Ref.~\cite{sannomiya2016supersymmetry}, 
and (c) the $\mathbb{Z}_2$-Nicolai model \cite{sannomiya2017supersymmetry}.  
In all these examples, supercharges consist of local terms shown by cyan triangles [living only on upward triangles in (a) and (b), and on both
upward and downward in (c)].  
Each local term is constructed by coupling a single fermion to bosonic operators living on particular bonds (highlighted in red).   
In (a) and (b), two-sublattice structures are automatically introduced by the construction, while all the sites are equivalent in (c).  
Pair-hopping process (arrows) in $H_{\rmm{Nic}}$ is also shown in (a).  
\label{def-Q}}
\end{figure}
%%%%%%%%%%%%%%%%%%%%%%%%%%%%%%%%%%%%%%%%%%%%%%%%%%%%%%%%%%%

%%%%%%%%%%%%%%%%%%%%%%%%%%%%%%%%%%%%%%%%%%%%%%%%%%%%%%%
\subsection{Symmetries and conserved quantities}
%%%%%%%%%%%%%%%%%%%%%%%%%%%%%%%%%%%%%%%%%%%%%%%%%%%%%%%
The Hamiltonian $H_{\rmm{S}} \qty(g)$ on a periodic lattice with $2L$ sites possesses the following symmetries.  
First of all, the nilpotency condition $Q_{\rmm{S}} \qty(g)^2=0$ immediately implies the supersymmetry of the model:
\begin{align}
    \com{H_{\rmm{S}} \qty(g)}{Q_{\rmm{S}} \qty(g)}=\com{H_{\rmm{S}} \qty(g)}{\hc{Q}_{\rmm{S}}\qty(g)}=0.
\end{align}
Also, from the construction \eqref{BCS-like} of the Hamiltonian $H_{\rmm{S}} \qty(g)$,
the energies are non-negative, and a zero-energy ground state $\ket{0}_0$ (throughout this paper, we use the symbol $\ket{0}$  
for ground states and the suffix $0$ to indicate that the state has zero energy) satisfying $H_{\rmm{S}} \qty(g)\ket{0}_0=0$, 
if it exists, must be a singlet that is annihilated by the supercharge: $Q_{\rmm{S}} \qty(g) \ket{0}_0=\hc{Q}_{\rmm{S}} \qty(g) \ket{0}_0=0$, and vice versa.  
All the other states with (strictly) positive energies appear as doublets \cite{witten1982constraints}.

 The usual global $U(1)_{c}$ symmetry that acts on $c_i$ as $c_i\rightarrow \eno{\im \theta}c_i \qty(\theta \in \mathbb{R})$ does {\em not} exist 
 except at $g=0$ where the model trivializes and $g=\infty$ where the system reduces to the Nicolai model.
 Instead, the system conserves the fermion parity of $c_i$, which is defined by $\qty(-1)^{N_c}=\qty(-1)^{\sum_i n_i^c}$ 
 and generates $\mathbb{Z}_{2}$ transformation: $\qty(-1)^{N_c} c_{i} \qty(-1)^{N_c}= -c_{i}$.  
 The Hamiltonian satisfies the condition $\com{H_{\rmm{S}} \qty(g)}{(-1)^{N_c}}=0$.  
 Moreover, the particular form of the pairing term $H_{\rmm{pair}}$ leads to the conservation of the following quantity:
 \begin{equation}
 \sum_{i} (-1)^{i} n_{i}^{c} \; ,
 \end{equation}
 which generates ``staggered'' gauge transformation $c_i\rightarrow \eno{(-1)^{i} \im \theta}c_i$.  
 Note that this is special to the nearest-neighbor pairing.

On the other hand, since the supercharge $Q_{\rmm{S}} \qty(g)$ [$Q^{\dagger}_{\rmm{S}} \qty(g)$] annihilates (creates) 
 precisely one $d$ fermion, the Hamiltonian $H_{\rmm{S}} \qty(g) =\acom{Q_{\rmm{S}} \qty(g)}{\hc{Q}_{\rmm{S}} \qty(g)}$ that contains only 
 the combinations $d_{i} d_{j}^{\dagger}$ and $d_{i}^{\dagger}d_{j}$  
 preserves the number of $d$ fermions and has $U(1)_d$ symmetry: $d_i\rightarrow \eno{i\theta}d_i\qty(\theta \in \mathbb{R})$.  

On a periodic lattice, the model (\ref{BCS-like}) has  the translation symmetry $\mathcal{T}$ that simultaneously shifts $c_i$ and $d_i$ by one site 
$c_i \xrightarrow{\mathcal{T}} c_{i+1}$ and $d_i \xrightarrow{\mathcal{T}} d_{i+1}$ is present.
 Furthermore, the system has an inversion symmetry $\mathcal{I}$ that acts as 
$c_i \xrightarrow{\mathcal{I}} \pm i c_{L+1-i \pmod L}$ and 
$d_i \xrightarrow{\mathcal{I}}  d_{L-i \pmod L}$.

%%%%%%%%%%%%%%%%%%%%%%%%%%%%%%%%%%%%%%%%%%%%%%%%%%%%%%%
\section{SPONTANEOUS SUPERSYMMETRY BREAKING}
\label{sec:SUSYSSB}
%%%%%%%%%%%%%%%%%%%%%%%%%%%%%%%%%%%%%%%%%%%%%%%%%%%%%%%
In this section, we prove the existence of spontaneous symmetry breaking (SSB)
of supersymmetry (SUSY) in the supersymmetric lattice model \eqref{BCS-like} and determine the critical value $g_{\text{c}}$ below which 
SUSY is broken.  
%As we will see below, for any finite system sizes $L$, SUSY is {\em always} spontaneously broken in $H_{\rmm{S}} \qty(g)$ as far as $g<\infty$.  
As the situation is subtle in the thermodynamic limit, we need to be careful in judging whether or not SUSY is broken in infinite-size systems.  
%%%%%%%%%%%%%%%%%%%%%%%%%%%%%%%%%%%%%%%%%%%%%%%%%%%%%%%
\subsection{Definition of the spontaneous SUSY breaking}
\label{sec:def-SUSY-SSB}
%%%%%%%%%%%%%%%%%%%%%%%%%%%%%%%%%%%%%%%%%%%%%%%%%%%%%%%
As has been emphasized in Ref.~\cite{witten1981dynamical}, the spontaneous breaking of SUSY is rather different from that of ordinary 
global symmetry.  
In fact, for {\em any} finite system size $L$ and finite coupling $g$, one can easily show that SUSY is always broken spontaneously in the ground state
of \eqref{BCS-like} in the sense that the ground-state energy is strictly positive and that all eigenstates are SUSY doublets 
\cite{witten1982constraints,cooper1995supersymmetry}.   

To show this, we consider the following local operator:
\begin{equation}
\chi_j\qty(g) := \left(1 - g \, \hc{c}_{j}\hc{c}_{j+1}\right) \hc{d}_{j} 
\label{eqn:def-conj-O}
\end{equation}
which satisfies 
\begin{equation}
\left\{Q_{\rmm{S}} \qty(g), \chi_j\qty(g)\right\} = 1
\label{eqn:acomm-Q-O}
\end{equation}
for any {\em finite} $g$.  
Such a local operator, if exists, has a significant meaning to the spontaneous breaking of SUSY.    
Let us first suppose that there exists a zero-energy ground state $\ket{0}_{0}$ satisfying 
$Q_{\rmm{S}} \qty(g) \ket{0}_{0} = \hc{Q}_{\rmm{S}} \qty(g) \ket{0}_{0} = 0$.  
Then, the following anticommutator must vanish:
\begin{equation}
{}_{0} \! \bra{0}\left\{Q_{\rmm{S}} \qty(g), \chi_j\qty(g)\right\} \ket{0}_{0} = 0\; .
\label{eqn:GS-average-acomm-1}
\end{equation}
On the other hand, from \eqref{eqn:acomm-Q-O}, we obtain ${}_{0}\!\bra{0} \left\{Q_{\rmm{S}} \qty(g), \chi_j\qty(g)\right\} \ket{0}_{0} = 1$ 
which contradicts with \eqref{eqn:GS-average-acomm-1}.   
Therefore, there is no zero-energy ground state and SUSY SSB always occurs in $H_{\rmm{S}} \qty(g)$ as long as $g$ and $L$ are finite.  

Note that the relation \eqref{eqn:acomm-Q-O} does not hold for $g = \infty$, 
where we must use $\lim_{g \to \infty}Q_{\rmm{S}} \qty(g)/g$ and $\lim_{g \to \infty} \chi_j\qty(g) /g$ instead,  
and hence Eq.~\eqref{eqn:acomm-Q-O} reads: $\lim_{g\to \infty} \left\{Q_{\rmm{S}} \qty(g)/g , \chi_j \qty(g)/g \right\} = 0$.   
Now it does not lead to any contradiction to \eqref{eqn:GS-average-acomm-1} and the above strategy no longer works in the limit $g \to \infty$.     
In fact, it is known \cite{moriya2018ergodicity,Katsura_2020} 
that SUSY remains unbroken in the Nicolai model with $\lim_{g \to \infty} H_{\rmm{S}} \qty(g)/g^{2}$.

In the above, we have seen that, except at $g=\infty$, SUSY is always broken in the ground state of $H_{\rmm{S}}(g)$ for any finite 
system size $L$.   
However, the situation is different in infinite-size systems where the ground-state energy itself may not be well-defined.
In fact, showing that $E_{\text{g.s.}}(g;L) >0$ for any finite $L$ is {\em not} sufficient to conclude that SUSY is broken in the limit $L \to \infty$ 
\cite{witten1982constraints}, 
and it is suggested that we should use an alternative definition of SUSY SSB in terms of the ground-state energy {\em density}.  
More concretely, we say that SUSY is spontaneously broken in a system of length $L$ (including the thermodynamic limit $L \to \infty$)  
when the ground-state energy $E_{\text{g.s.}}(g;L)$ {\em per site} satisfies 
\cite{sannomiya2016supersymmetry,sannomiya2017supersymmetry,sannomiya2019supersymmetry}:
\begin{equation}
\frac{E_{\text{g.s.}}(g;L)}{L}
%=\lim_{L\rightarrow \infty} \frac{\bra{\rmm{g.s.}}H\ket{\rmm{g.s.}}}{L} 
\gneq 0   \; .
\label{eqn:SSB-criterion-gen}
\end{equation}
For any finite $L$, this reduces to the usual criterion $E_{\text{g.s.}}(g;L) \gneq 0$ \cite{witten1982constraints} 
and, in the infinite-size limit $L \to \infty$, this is equivalent to strict positivity of the ground-state energy density:
\begin{equation}
e_{\text{g.s.}}(g) := 
\lim_{L\rightarrow \infty} \frac{E_{\text{g.s.}}(g;L)}{L} \gneq 0   \; .  
\label{eqn:SSB-criterion}
\end{equation}
According to the definition \eqref{eqn:SSB-criterion-gen}, for SUSY to be broken even in the infinite-size limit,  
the positive ground state energy $E_{\text{g.s.}}(g;L)$ must diverge as $L$ or faster.  
In what follows, we will consider infinite-size systems unless otherwise stated and use Eq.~\eqref{eqn:SSB-criterion} to judge SUSY SSB.

Given the above criterion for SUSY SSB, we can immediately conclude that SUSY is broken at least when $g=0$ since
$e_{\text{g.s.}}(g=0) = H_{\rmm{S}} \qty(g=0)/L=1$ there.  
Then, the next question is whether SUSY remains broken even for finite $g$ or not.  
When we increase $g$ from zero, the pairing fluctuations in $g H_{\rmm{pair}}$ tend to lower the ground-state energy density $e_{\text{g.s.}}(g)$ 
as $- 2g/\pi$ [see Eq.~\eqref{eqn:Egs-density}].  
Therefore, if we ignore the $g^{2}$ correction from $g^{2} H_{\rmm{Nic}}$, $e_{\text{g.s.}}(g)$ stays positive when $g$ is small enough and finally reaches zero at $g=\pi/2$, above which SUSY is expected to be restored. 
However, in Sec.~\ref{sec:proof-for-infinite-L}, we will show that, as in the case of finite systems, the presence of operator $\chi_i \qty(g)$ 
\eqref{eqn:def-conj-O} leads to SUSY SSB for any finite values of $g$.

%%%%%%%%%%%%%%%%%%%%%%%%%%%%%%%%%%%%%%%%%%%%%%%%%%
\subsection{Superconducting order parameter and SUSY breaking}
%%%%%%%%%%%%%%%%%%%%%%%%%%%%%%%%%%%%%%%%%%%%%%%%%%
\label{Superconducting order parameter and SUSY breaking}
Before proceeding to the discussion of infinite-size systems, let us consider SUSY breaking from the order-parameter 
point of view.  To this end, we consider the following bosonic  quantities, defined through infinitesimal supersymmetric transformations 
of fermionic operators \cite{wess1992supersymmetry}:
\begin{subequations}
\begin{align}
F_i: &=\acom{Q_{\text{S}}(g)}{c_i}
=g \qty(
\hc{c}_{i+1}d_{i}-\hc{c}_{i-1}d_{i-1} ),    
\label{eqn:OP-F-by-fermions}    \\
G_i:&=\acom{\hc{Q}_{\text{S}}(g)}{d_i}
=1 - g\,  c_i c_{i+1}.
\label{eqn:OP-G-by-fermions}
\end{align}
\end{subequations}
That these work as the order parameters can be understood as follows. 
First of all, in a phase where SUSY is unbroken and a zero-energy ground state $\ket{0}_{0}$ exists, 
${}_{0} \! \bra{0} F_i \ket{0}_{0} = {}_{0} \! \bra{0} G_i \ket{0}_{0} = 0$ follows automatically from 
$Q_{\rmm{S}} \qty(g) \ket{0}_{0} = \hc{Q}_{\rmm{S}} \qty(g) \ket{0}_{0} = 0$.   
This implies that if at least one of ${}_{0} \! \bra{0} F_i \ket{0}_{0}$ and ${}_{0} \! \bra{0} G_i \ket{0}_{0}$ takes a  non-zero value, SUSY must be broken. 
These order parameters can be introduced systematically through the use of superfield formalism: see Appendix \ref{App:chobakeisiki} for more details.  

The order parameter $\bra{0} F_i \ket{0}$ \eqref{eqn:OP-F-by-fermions} has a physical interpretation as a current-like quantity 
flowing along the zigzag direction (see Fig.~\ref{def-Q}), 
and it changes the sign under the inversion symmetry $\mathcal{I}$.  
Numerical calculations 
of the expectation value of $F_i$ yielded values close to zero, indicating that the ground state 
does not break the inversion symmetry $\mathcal{I}$.  

Next we calculate the order parameter $\bra{0} G_i \ket{0}$ for the model $H_{\rmm{S}} \qty(g)$.  
In the SUSY-unbroken phase $g > g_{\text{c}}$, if it exists, ${}_{0}\! \bra{0} G_i \ket{0}_{0} = 0$ and 
the superconducting order parameter $\Delta(g):=\langle 0 | c_i c_{i+1}|0 \rangle$ should be given {\em exactly} by $\Delta(g)=\frac{1}{g}$; 
any deviation from this value, no matter how small it is, implies the broken SUSY in finite size systems.   
To calculate $\bra{0} G_i \ket{0}$ in the SUSY-broken phase ($g < g_{\text{c}}$), we consider an extreme case with $g\ll 1$ 
in which the ground state is given approximately by that of $H_{\rmm{pair}}$.   
Then, with the method used in Sec.~\ref{sec:1st-order-in-Hnic}, 
the superconducting order parameter and $\bra{0} G_i \ket{0}$ are easily calculated as:
\begin{equation}
\begin{split}
&    \Delta(g) ={}_{c} \langle 0 |\hc{U}_c c_ic_{i+1}U_c|0 \rangle_{c} =\frac{1}{L}\sum_{k>0}\sin(k)    \\
& \bra{0} G_i \ket{0} = 1- g \Delta(g)   \; .
\end{split}
\label{eqn:pairing-vs-SUSY-G}
\end{equation}
In particular, this reduces to $\Delta(g)=1/\pi$ in the infinite-size limit $L\rightarrow\infty$.  

To obtain the value of $\Delta(g)$ for generic $g$, we numerically calculated it for a finite system with $L=6$ (12 sites) 
under periodic boundary conditions (see Fig.~\ref{gap-func-Numerical}).  
We observed that, when $g \approx 0$, it agrees with the $g \ll 1$ value $\frac{1}{L}\sum_{k>0}\sin(k)=\frac{\sqrt{3}}{6}=0.288675\dots$ 
and approaches $1/g$ from below.  
In the inset of Fig.~\ref{gap-func-Numerical}, the difference $1/g - \Delta(g)=\bra{0} G_i \ket{0}/g$ is plotted which, according to the general argument 
presented above, signals SUSY breaking (i.e., any finite values of it imply broken SUSY).   
The data were well fitted by $1/g - \Delta(g)=\alpha + \beta/g^{\gamma}$ with $\alpha=-1.35\times 10^{-7}$, $\beta=6.74$, and $\gamma=2.98 (\approx 3)$ 
(see the red dashed line in the inset) \footnote{%
This is natural from the symmetry of $\Delta(g)$. As the sign change $g \to -g$ is compensated by $c_i\migi \eno{\im \frac{\pi}{2}}c_i$, the superconducting order parameter is an odd function of $g$: $\Delta(-g)=- \Delta(g)$. Therefore, if $\Delta(g)$ allows a $1/g$-expansion, $\alpha=0$ and it must read as: $\Delta(g)=1/g + \beta/g^{3} + \cdots$.}.   
If we neglect the very tiny offset $\alpha$, this is consistent with the observation in Sec.~\ref{sec:def-SUSY-SSB} that SUSY is broken 
for any finite $g$ (i.e., $g_{\text{c}}=\infty$) when the system size $L$ is finite.
Therefore, $\Delta(g)$ just approaches $1/g$ from below but they never intersect (if they do at some $g$, it means that SUSY is restored there).   

%%%%%%%% FIG %%%%%%%%%%%%%%%%%%%%%%%%%%%%%%%%%%%%%
\begin{figure}[htbp]
\includegraphics[width=\columnwidth,clip]{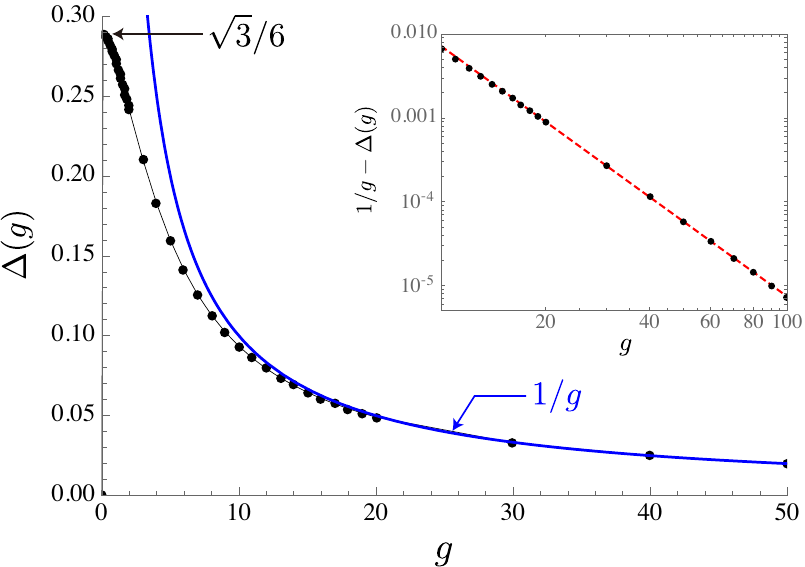}
\caption{The $g$-dependence of the gap function $\Delta(g):=\langle 0 | c_i c_{i+1}|0 \rangle$ (black dots) obtained 
by numerical diagonalization of a finite-size system ($L=6$, i.e., 12 sites).  The order parameter $\bra{0} G_i \ket{0}/g = 1/g - \Delta(g)$ is shown for large-$g$ in the inset (in log-log scale). Note that for large $g$ (e.g., $g \gtrsim 30$), $\Delta(g)$ behaves like $- \alpha + 1/g - \beta/g^{\gamma}$ 
with $\alpha=-1.35\times 10^{-7}$, $\beta=6.74$, and $\gamma=2.98 (\approx 3)$ (red dashied line). 
\label{gap-func-Numerical}
}
\end{figure}
%%%%%%%%%%%%%%%%%%%%%%%%%%%%%%%%%%%%%%%%%%%%%%%

%%%%%%%%%%%%%%%%%%%%%%%%%%%%%%%%%%%%%%%%%%%%%%%%%%%%
\subsection{SUSY breaking in infinite-size systems}
\label{sec:proof-for-infinite-L}
%%%%%%%%%%%%%%%%%%%%%%%%%%%%%%%%%%%%%%%%%%%%%%%%%%%%
While we have seen in the previous subsections that SUSY is always broken except at the point $g=\infty$ (the Nicolai model) 
as far as the system size is finite, the situation is subtle in infinite-size systems.  
In this respect, a thermodynamic-limit statement has been obtained \cite{moriya2018supersymmetry} for the extended Nicolai model introduced 
in Ref.~\cite{sannomiya2016supersymmetry}.  
Following Ref.~\cite{moriya2018supersymmetry}, we will prove in this subsection that the transition point of $H_{\rmm{S}}(g)$ 
(in the thermodynamic limit) 
is in fact $g_\rmm{c}=\infty$, i.e., SUSY remains broken in the sense of Eq.~\eqref{eqn:SSB-criterion} all the way up to $g=\infty$ 
\footnote{%
In Ref.~\cite{moriya2018supersymmetry}, another definition of SUSY breaking is adopted to prove broken SUSY in the infinite-size limit.  
Here we use the strategy used to prove the positivity of the energy density.}.   

%In the following, we consider infinite-size systems. 
Let us denote the normalized ground state of the system and its energy by $\ket{0}$ and $E_{\text{g.s.}}$, 
respectively.  Then, the following lemma holds. 
%%%%%%%%%%%%%%%%%%
\begin{lem}
If we assume that SUSY breaking does not occur in the sense of Eq.~\eqref{eqn:SSB-criterion}, 
then both $\norm{Q\ket{0}}/\sqrt{L}$ and $\norm{\hc{Q}\ket{0}}/\sqrt{L}$ approach zero as $L\to \infty$, i.e., 
%we have:
\begin{equation}
\label{Qorder}
\norm{Q\ket{0}}  \, , \;  \norm{\hc{Q}\ket{0}} = o\qty(\sqrt{L}) \; .
\end{equation}  
\label{thm:lemma-1}
\end{lem}
%%%%%%%%%%%%%%%%%%
\begin{proof}
When SUSY is not broken, Eq.~\eqref{eqn:SSB-criterion} implies:
\begin{align}
\frac{E_{\text{g.s.}}}{L}=
\frac{
\bra{0}\acom{Q}{\hc{Q}}\ket{0}
}{L} \nonumber
=\frac{\norm{Q\ket{0}}^2}{L}+\frac{\norm{\hc{Q}\ket{0}}^2}{L} \rightarrow 0  \; .
\end{align}
Since each term is positive, we immediately see:
\begin{align}
&\norm{Q\ket{0}}^2\, ,  \; \norm{\hc{Q}\ket{0}}^2=o\qty(L)
\end{align}
which, after taking square roots, gives the desired results.  
\end{proof}

In proving SUSY-breaking for {\em finite} $L$, the operator $\chi_{i}(g)$ defined by Eq.~\eqref{eqn:def-conj-O} and 
the anticommutation relation \eqref{eqn:acomm-Q-O} have played a crucial role.   
In the same spirit, we define the following operator $X \qty(g)$:
\begin{equation}
X \qty(g) := \frac{1}{L}\sum_{i=1}^{L} \chi_i\qty(g), 
\label{eqn:def-X}
\end{equation}
which satisfies 
\begin{equation}
\label{conjugateQX}
\acom{Q_{\rmm{S}}\qty(g)} {X \qty(g)}=1 
\end{equation}
{\em provided that $g < \infty$}.   
Below, we assume $g < \infty$ and suppress the $g$-dependence of $X(g)$ for simplicity of the notation.   

Then, we can show that the operators $X$ and $X^{\dagger}$ have the following system-size dependence due to the short-range nature of 
$\chi_{i}(g)$.  

%%%%%%%%%%%%%%%%
\begin{lem}
    $\norm{X \ket{0}},\norm{\hc{X}\ket{0}}=O\qty(1/\sqrt{L})$.
    \label{thm:lemma-2}
\end{lem}
%%%%%%%%%%%%%%%%
\begin{proof}
First, we note the following inequality 
    \begin{equation}
\label{chiineq}
\begin{split}
    \norm{X \ket{0}} &\leq 
    \sqrt{\norm{X \ket{0}}^2+\norm{\hc{X}\ket{0}}^2}  \\
    &=\sqrt{\bra{0}\acom{X}{\hc{X}}\ket{0}} \\
    &=\frac{1}{L}\sqrt{
    \sum_{i,j}\bra{0}\acom{\chi_i\qty(g)}{\hc{\chi}_j\qty(g)}\ket{0}
    }
\end{split}
\end{equation}
that can be readily derived using the definition \eqref{eqn:def-X} of $X$.  
A similar inequality holds for $\norm{\hc{X}\ket{0}}$ as well.  

To evaluate the right-hand side, we note $\acom{\chi_i\qty(g)}{\hc{\chi}_j\qty(g)}=0$ for $\abs{i-j}\geq 2$.  
Then, assuming translationally invariant $\ket{0}$, the quantity inside the square root is of order $L$ and we obtain:
\begin{equation}
\label{chichiineq}
% \norm{X \ket{0}} \leq 
    \frac{1}{L}\sqrt{
    \sum_{i,j}\bra{0}\acom{\chi_i\qty(g)}{\hc{\chi}_j\qty(g)}\ket{0}
    }
=O\qty(1/\sqrt{L})  \; .
\end{equation}
Therefore, $\norm{X \ket{0}}$ and $\norm{\hc{X}\ket{0}}$ behave like:
\begin{equation}
    \norm{X \ket{0}}\, , \;   \norm{\hc{X}\ket{0}} = O\qty(1/\sqrt{L})  \; .
\end{equation}
\end{proof}

Combining these two Lemmas with Eq.~\eqref{conjugateQX}, we finally arrive at the following statement. 
%%%%%%%%%%%%%%%%
\begin{thm}
    Supersymmetry is broken in the thermodynamic limit for any $g<\infty$, i.e., $g_\rmm{c}=\infty$.
\end{thm}
%%%%%%%%%%%%%%%%
\begin{proof}
From \eqref{conjugateQX} and the Cauchy-Schwarz inequality, we obtain the following:
\begin{equation}
\begin{split}
    1&=\bra{0}\acom{Q}{X}\ket{0}=\bra{0}\qty(Q X + X Q )\ket{0}   \\
    &\leq \norm{\hc{Q}\ket{0}}\norm{X \ket{0}}+\norm{\hc{X}\ket{0}}\norm{Q\ket{0}}  \; .
\end{split}
\label{eqn:acomm-plus-CS-ineq}
\end{equation}
Now let us assume that SUSY breaking does not occur for $g<\infty$.  
Then, combining the above with Lemmas \ref{thm:lemma-1} and \ref{thm:lemma-2}, we see that 
the right-hand side of Eq.~\eqref{eqn:acomm-plus-CS-ineq} vanishes in the limit $L \to \infty$,  
%we have 
%\remark{Check!!} 
%\begin{align}
%    1&\leq \qty(\norm{Q\ket{0}}+\norm{\hc{Q}\ket{0}})
%    \sqrt{\bra{0}\acom{\chi}{\hc{\chi}}\ket{0}}
%\end{align}
%Plugging Eqs.~\eqref{Qorder} and \eqref{chichiineq} into the above, we obtain
%\begin{equation}
%1 \leq 
% \norm{\hc{Q}\ket{0}}\norm{X \ket{0}}+\norm{\hc{X}\ket{0}}\norm{Q\ket{0}} 
% = o\qty(1/\sqrt{L})   \; ,
%\end{equation}
which is a contradiction.
\end{proof}

Thus, we have established that SUSY is broken for any $g < \infty$ regardless of the system size $L$. 
Considering the close relation \eqref{eqn:pairing-vs-SUSY-G} between the pairing amplitude $\Delta(g)$ 
and the SUSY-breaking order parameter $G$, this seems quite natural. 
Any finite $g$ explicitly introduces fermion pairing that makes $\Delta(g)$ have a generic finite value 
leading to finite $G$; in order for $G=0$, the pairing amplitude $\Delta(g)$ needs to take the fine-tuned value $1/g$ which is unlikely 
for generic values of $g$. 
%%%%%%%% FIG %%%%%%%%%%%%%%%%%%%%%%%%%%%%%%%%%%%%%
\begin{figure}[htbp]
\includegraphics[width=\columnwidth,clip]{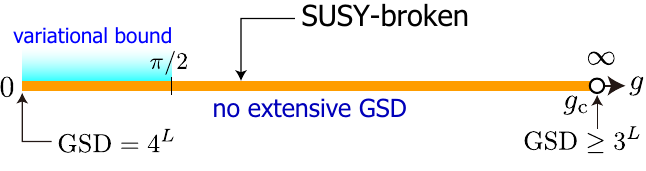}
\caption{
The phase diagram of the SUSY lattice model $H_{\rmm{S}} \qty(g)$ defined in Eq.\eqref{BCS-like}. The existence of the operator $X\qty(g)$ conjugate to $Q_\rmm{S}\qty(g)$ enables us to prove that SUSY is always broken for any finite $g$ (both in finite and infinite systems), i.e., 
$g_{\text{c}}=\infty$.  
A simple variational argument in Appendix \ref{sec:variational} tells that SUSY is broken at least for $g < \pi/2$ (i.e., $g_\rmm{c}\geq\pi/2$).   
The values of ground-state degeneracy (GSD) at $g=0$ and $g=\infty$ are also indicated (the value at $g=\infty$ are given  
in Refs.~\cite{moriya2018ergodicity, Katsura_2020}). Numerical and analytical results in Sec.~\ref{sec:SF} suggest that 
the extensive ground-state degeneracy is lifted (aside from $O(1)$ degeneracy associated with zero modes) for finite $g$.}
\end{figure}
%%%%%%%%%%%%%%%%%%%%%%%%%%%%%%%%%%%%%%%%%%%%%%%

%%%%%%%%%%%%%%%%%%%%%%%%%%%%%%%%%%%%%%%%%%%%%%%%%%%%
\subsection{Dispersion of NG fermions}
%%%%%%%%%%%%%%%%%%%%%%%%%%%%%%%%%%%%%%%%%%%%%%%%%%%%
Using a variational argument, 
we can establish a bound for the NG fermion dispersion, confirming the presence of gapless excitations.   
In the case of ordinary broken symmetries, the variational state for the corresponding NG {\em boson} is created by 
applying the generators of the broken symmetries to the ground state (the single-mode approximation).  
In the same spirit, we may create a trial state for the NG fermion by replacing the bosonic generators with the fermionic supercharge.
Specifically, following Refs.~\cite{sannomiya2016supersymmetry,sannomiya2017supersymmetry,sannomiya2019supersymmetry}, 
we consider the local supercharge operators
\begin{align}
q_i := (1+g \hc{c}_i\hc{c}_{i+1})d_i 
\end{align}
and introduce its momentum-$p$ component as:
\begin{align}
Q_p := \sum_{j =1}^L \eno{\im p j}q_j.
\end{align}
We then define the trial excited state with momentum $p$ as
\begin{align}
\ket{\psi_p} := \frac{\qty(Q_p+\hc{Q}_p)\ket{0}}{\norm{\qty(Q_p+\hc{Q}_p)\ket{0}}},\ \ \ \qty(p\neq 0),
\end{align}
where $\ket{0}$ is the normalized ground state.  
Then, the variational energy (the Feynman frequency) of this trial state $\ket{\psi_p}$ can be evaluated 
with the help of the inversion symmetry and inequalities given  
in Refs.~\cite{sannomiya2016supersymmetry,sannomiya2017supersymmetry,sannomiya2019supersymmetry} as: 
\begin{align}
\eps_{\rmm{var}}\qty(p) :&= \bra{\psi_p}H_S\ket{\psi_p}-E_{\text{g.s.}} \nt
&= \frac{\bra{0}\com{Q_p}{\com{H_S}{\hc{Q}_p}}\ket{0}}{\bra{0}\acom{Q_p}{\hc{Q}_p}\ket{0}} \nt
&\leq \sqrt{\frac{\bra{0}\com{\com{Q_p}{H_S}}{\com{H_S}{\hc{Q}_p}}\ket{0}}{\bra{0}\acom{Q_p}{\hc{Q}_p}\ket{0}}} \nt
&= \sqrt{\frac{C}{E_{\text{g.s.}}/L}}\abs{p}+O\qty(\abs{p}^3),
\end{align}
%\remark{Q: can we use this argument even when there are several different G.S. in the thermodynamic limit??} 
where $C$ is a positive constant independent of $L$, and $E_{\text{g.s.}}/L \, (\gneq 0)$ is the ground-state-energy density.
%%%%%%%%%%%%%%%%%%%%%%%%%%%%%%%%%%%%%%%%%%%%%%%%%%%
This variational energy provides an upper bound on the true dispersion of the Nambu-Goldstone mode due to the sum rule.
Here, it should be noted that this variational approach is not effective when the trial state $\ket{\psi_p}$ coincides 
with other orthogonal ground states $\ket{0}$.

A more systematic approach may be developed based on the non-linear representation of SUSY \cite{volkov1973neutrino} 
using the specific form of the operator $Q_S\qty(g)$ and its conjugate $O_j\qty(g)$ introduced 
in Sec.~\ref{sec:def-SUSY-SSB} [see Eq.~\eqref{eqn:def-conj-O}].  
It would be a very interesting open problem to pursue this direction and construct a systematic theory for
low-energy {\em fermionic} excitations in supersymmetric lattice models.

%%%%%%%%%%%%%%%%%%%%%%%%%%%%%%%%%%%%%%%%%%%%
\section{Extensive ground-state degeneracy and Flat band}
\label{sec:SF}
%%%%%%%%%%%%%%%%%%%%%%%%%%%%%%%%%%%%%%%%%%%%
In the trivial case $H_\text{S} (g = 0) = L$, there is $2^{L} \times 2^{L}$-fold ground-state degeneracy 
associated with the complete decoupling of all the fermionic degrees of freedom from the Hamiltonian. 
In the case of the Nicolai model (at $g=\infty$), on the other hand, it is known that the ground state is at least $3^L$-fold degenerate \cite{moriya2018ergodicity,Katsura_2020}.     
Such extensive ground-state degeneracy is thought of as a common feature in many supersymmetric lattice models 
and is dubbed superfrustration 
\cite{huijse2008superfrustration,huijse2012supersymmetric,PhysRevLett.101.146406,PhysRevLett.95.046403} 
by analogy to similar phenomena in frustrated magnetism.  

At the first order in $g$, the $c$-fermions develop finite dispersion and this $2^{L} \times 2^{L}$-fold degeneracy is partially lifted 
down to $O(2^{L})$ as we will see in Sec.~\ref{sec:numerical-spec}.   
Up to this order, the $d$-fermions do not appear in the Hamiltonian and constitute a zero-energy flat band which is responsible 
for the extensive ground-state degeneracy.  
In this section, we first numerically investigate the evolution of the energy spectrum in $g$ to see if the extensive ground-state degeneracy 
and the zero-energy flat band persists for finite $g$ or not.  
Then, we try to understand the spectral structure, especially the fate of the zero-energy flat band, by a perturbative argument.  

%%%%%%%%%%%%%%%%%%%%%%%%%%%%%%%%%%%%%%%%%%%%%%%%%%%%%%%%%%%%%%%%%%%%
\subsection{Numerical spectrum}
\label{sec:numerical-spec}
%%%%%%%%%%%%%%%%%%%%%%%%%%%%%%%%%%%%%%%%%%%%%%%%%%%%%%%%%%%%%%%
It is useful to numerically calculate the eigenvalues of the Hamiltonian \eqref{BCS-like}.   
In Fig.~\ref{fig:spec-str}, all the eigenvalues of the finite-size Hamiltonian ($L=6$, 12 sites, $2^{12}=4096$ levels)  
are plotted for various values of $g$ 
($=0.01$, $0.1$, $0.2$, $1.0$, $5.0$, and $50.0$).     
The five-step structure for smallest $g=0.01$ [Fig.~\ref{fig:spec-str}(a)] is easily understood by 
keeping only the zero-energy $d$-band and the pairing part $g H_{\text{pair}}= 2g \sum_{k} |\sin k| \qty(f_{k}^{\dagger}f_{k} - \hf)$ 
[$k=0$, $\pm \pi/3$, $\pm 2\pi/3$, and $\pi$; see Eq.~\eqref{BGGtrf} for the definition of the $f$-fermion].   
The trivial $g^{0}$ energy $L \,(=6)$ [dashed line in Fig.~\ref{fig:spec-str}(a)] acquires the $O(g)$ correction 
$-g \sum_{k} |\sin k| = - 2\sqrt{3} g \,( = - 0.034641... \text{ at $g=0.01$})$  
to give the ground-state energy density up to the first order in $g$.  
For even-$L$, there are two zero-energy modes at $k=0$ and $\pi$ (when $L$ is odd, there is only one  
at $k=0$) and three states $f_{0}^{\dagger}|0\rangle_{f}$, $f_{\pi}^{\dagger}|0\rangle_{f}$, and $f_{0}^{\dagger}f_{\pi}^{\dagger} |0\rangle_{f}$ 
are degenerate with $|0\rangle_{c}$.    Neglecting the $g^{2}$ contributions, the $d$ fermions do not contribute the energy and any ($2^{L}$) 
$d$-fermion states have exactly zero energy.   Therefore, for sufficiently small $g$, the ground-state energy is 
$4 \times 2^{L} = 2^{8} = 256$-fold quasi-degenerate as is clearly seen in Fig.~\ref{fig:spec-str}(a).  
If we further increase $g$, the group of $4 \times 2^{L}$ quasi-degenerate levels acquires a finite slope [$\sim g^{2}$; compare the levels 
of the full Hamiltonian plotted by black points and those of the first-order Hamiltonian $L+g H_{\rmm{pair}}$ 
shown by red points in Figs.~\ref{fig:spec-str}(b) and (c)] 
which may suggest that the $4^{L}$-fold ground-state degeneracy that exists at $g=0$ is completely lifted for finite $g$.

The levels that constitute the second step are created by exciting one $f$-fermion to one of the four modes $k=\pm \pi/3$, $\pm 2\pi/3$ and have the energy 
(when $g=0.01$):
\begin{equation}
 E_{0} + 2g \sin (\pi/3)= E_{0} + \sqrt{3}g  \approx 5.98268...  
 \end{equation}
(the same energies for all the other cases) with 
 \[
 \begin{split}
 E_{0} :&= L- \frac{1}{2} \sum_{n=-L/2+1}^{L/2} 2g | \sin (2\pi n/L) | \\
&\approx 5.96536... 
\quad (L=6, \; g=0.01) \; .
\end{split} 
\] 
Since there are four different zero-mode occupations and four ($={}_{4} \text{C}_{1}$) different ways to excite a fermion 
to one of the four modes with finite energies, 
the second level is $4 \times 4 \times 2^{6} = 2^{10}=1024$-fold quasi-degenerate.  
We can repeat the same argument for higher levels as well.   For instance, in the third level, two $f$-fermions are excited 
(therefore, there is an energy increase by $2 \times \sqrt{3} g$).   We can create these state in six (${}_{4}\text{C}_{2}=6$) 
different ways, and, taking the zero modes into account, we have $6 \times 4 \times 2^{L} = 1536$-fold quasi-degeneracy.  

This picture is correct as long as $g$ is sufficiently small, i.e. when the $g^{2}$ part $H_{\rmm{Nic}}$ can be neglected.   
For larger $g$ (e.g., $g=0.1$), the $d$ fermions start interacting with $c$ through $H_{\text{Nic}}$.  Nevertheless, the gaps ($\sim g$) between 
adjacent levels are still robust against this $O(g^{2})$ effect [see Figs.~\ref{fig:spec-str}(b) and (c)]; 
the net outcome is that each step now has a finite slope ($\sim g^{2}$).   
The finite slopes imply that the $d$-band is no longer zero-energy flat; the band has at least a finite energy even if it is still flat.   

For larger $g\, (\gtrsim 1)$, the second term $g^{2} H_{\rmm{Nic}}$ that mixes the $c$ and $d$ fermions can no longer be neglected 
and the clear step-like structure which is indicative of a zero-energy flat band is smeared out [see Fig.~\ref{fig:spec-str}(d)].  
At around $g \gtrsim 5.0$, a new step appears around $E=0$ [see Figs.~\ref{fig:spec-str}(e) and (f)] signaling the zero-energy flat band 
which is known to exist \cite{Katsura_2020} in the Nicolai model.

%(in this respect, the mean-field prediction that the $d$-fermions are decoupled and remain zero-energy flat for arbitrary $g$ is not correct).  
%%%%%%%% FIG %%%%%%%%%%%%%%%%%%%%%%%%%%%%%%%%%%%%%
\begin{figure}[H]
\includegraphics[width=0.45\columnwidth,clip]{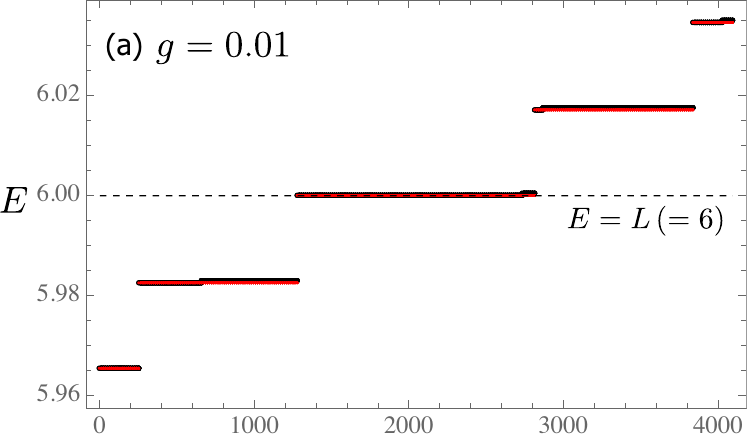} \hspace{2mm}
\includegraphics[width=0.45\columnwidth,clip]{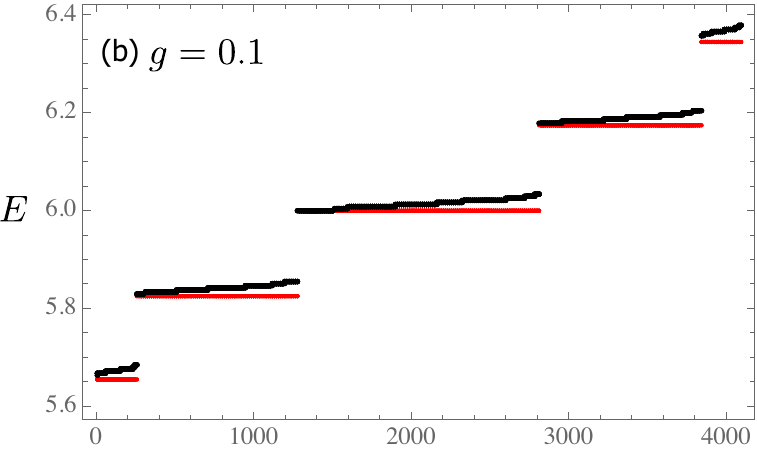} \hspace{2mm}
\includegraphics[width=0.45\columnwidth,clip]{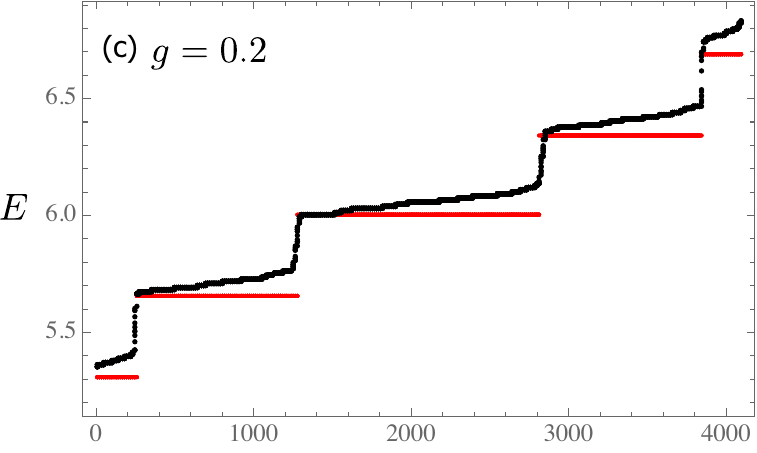} \hspace{2mm}
\includegraphics[width=0.45\columnwidth,clip]{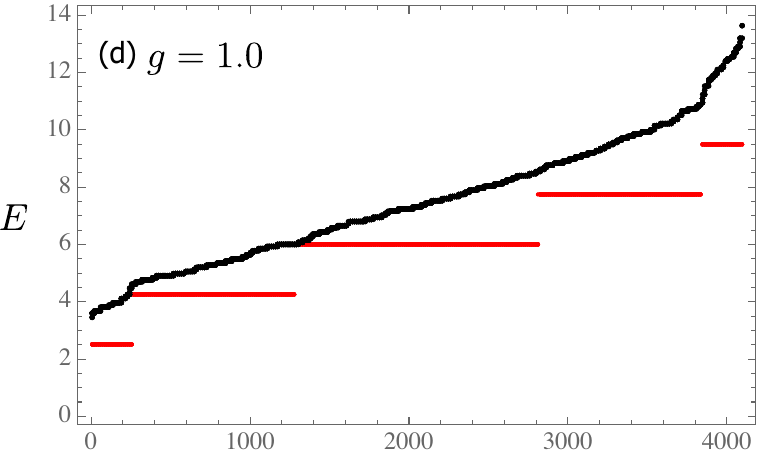} \hspace{2mm}
\includegraphics[width=0.45\columnwidth,clip]{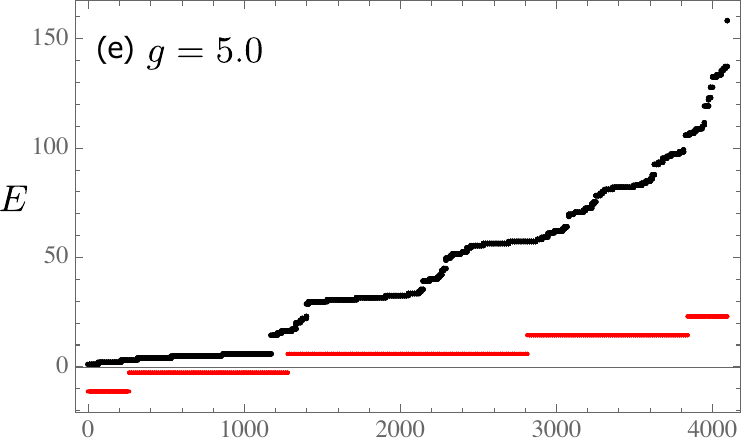} \hspace{5mm}
\includegraphics[width=0.45\columnwidth,clip]{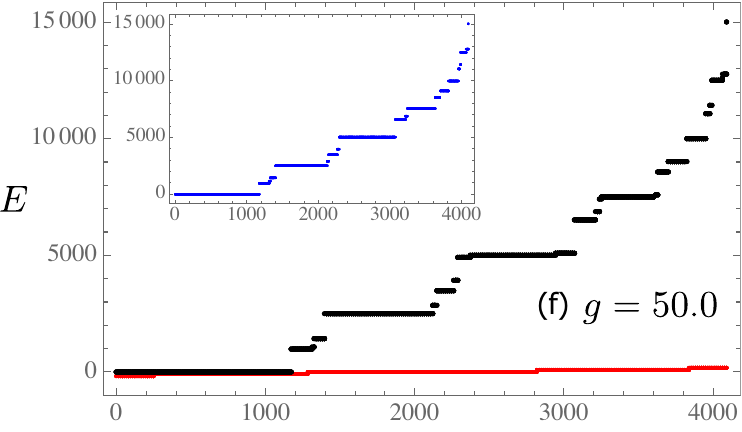} \hspace{2mm}
\caption{
All the eigenvalues (black dots) of the finite size ($L=6$, 12 sites) Hamiltonian $H_{\rmm{S}} \qty(g)$ 
for (a) $0.01$, (b) $0.1$, (c) $0.2$, (d) $1.0$, (e) $5.0$, and (f) $50.0$.  
The spectrum of the first-order partial Hamiltonian $L+g H_{\rmm{pair}}$ is plotted by red points.  
The system size is fixed to $L=6$ (i.e., $12$ sites, $2^{12}=4096$ levels).    
The five-step structure seen for $g=0.01$ is easily understood by taking into account 
the zero-energy flat $d$-band and the pairing part $g H_{\text{pair}}$ (see the text).   
For $g=1.0$ the structure looks smooth in contrast to the step-like structures 
for small $g$.  Between $g=1.0$ and $5.0$, reorganization of the levels occurs and the level structure changes essentially 
to that of the Nicolai model (see the inset of the last panel).  
\label{fig:spec-str}
}
\end{figure}
%%%%%%%%%%%%%%%%%%%%%%%%%%%%%%%%%%%%%%%%%%%%%%%
%%%%%%%%%%%%%%%%%%%%%%%%%%%%%%%%%%%%%%%%%%%%%%%%%%%%%%%%%%%%%%%%%%%%%%%%%S

%%%%%%%%%%%%%%%%%%%%%%%%%%%%%%%%%%%%%%%%%%%%%%%%%%%%%%%%%%%%%%%%%%%%
\subsection{Spectrum of $d$-fermions by first-order perturbation theory}
\label{sec:1st-order-in-Hnic}
%%%%%%%%%%%%%%%%%%%%%%%%%%%%%%%%%%%%%%%%%%%%%%%%%%%%%%%%%%%%%%%%%%%%
Since the $d$ fermion appears only in the interaction part $H_{\rm Nic}$ (the Nicolai model), the ground states of $L+g H_{\rmm{pair}}$ 
are $2^L$-fold degenerate with respect to the $d_i$ fermion configurations (as has been mentioned, since there is additional four-fold degeneracy 
associated with the $c$-fermions, the actual degeneracy is $2^{L+2}$).  
Assuming $g\ll1$, we can treat $H_{\rmm{Nic}}$ as a perturbation to $L+g H_{\rmm{pair}}$ and
 find the single-particle dispersion of the $d_i$ fermions by first-order degenerate perturbation (in $g^{2}$).

By the unitary transformation defined in Eq.~\eqref{Ubogo}:
\begin{equation}
U_{c} =\prod_{k>0} \exp[
\im \frac{\pi}{4}  \qty(c_{k}c_{-k}+\hc{c}_{-k}\hc{c}_{k})
]   \; ,
\end{equation}
the ground states of $H_{\mathrm{pair}}$ are written as 
$\ket{n}_{c+d} :=\ket{\text{g.s.}}_c \otimes \ket{\psi_{n}}_d$
with $\ket{\text{g.s.}}_c:= U_{c} \ket{0}_c \,(=\ket{0}_{f})$
and $\ket{\psi_{n}}_d$ standing respectively for the ground state of $H_{\mathrm{pair}}$ and any of the $2^L$ possible states of the $d_i$ fermion, 
which enables us to obtain the following matrix elements:

\begin{subequations}
\begin{align}
&\bra{\text{g.s.}} \hc{c}_{i}c_{j} \ket{\text{g.s.}}_c =\hf \delta_{i,j} \\
& \bra{\text{g.s.}} n^{c}_{i}n^{c}_{i+1} \ket{\text{g.s.}}_c =\frac{1}{4}+\frac{1}{\pi^2} \; .
\end{align}
\end{subequations}
Therefore, the effective Hamiltonian for the $d$-fermions within the subspace spanned by $\ket{n}_{c+d}$ is calculated as: 
\begin{equation}
%{}_{c}\!\bra{f}\otimes {}_{d}\!\bra{\psi_{n}} H_{\rmm{Nic}} \ket{f}_{c} \otimes \ket{\psi_{n^{\prime}}}_{d} =\frac{g^{2}}{4}L  \; .
\bra{n}g^2 H_{\rmm{Nic}} \ket{n^{\prime}}_{c+d} =\qty(\frac{1}{4}+\frac{1}{\pi^2})g^2L \, \delta_{n,n^{\prime}} \; ,
\end{equation}
which means that the net effect of the first-order perturbation $g^{2} H_{\rmm{Nic}}$ is just to shift the ground-state energy density 
by a constant $\qty(\frac{1}{4}+\frac{1}{\pi^2})g^{2}$.  
Therefore, as far as $g$ is sufficiently small, the flatness of the $d$-band and the $2^L$-fold degeneracy of the ground state remains.  

%%%%%%%%%%%%%%%%%%%%%%%%%%%%%%%%%%%%%%%%%%%%%%%%%%%%%%%%%%%%%%%%%%%%
\subsection{Second-order perturbation theory}
%%%%%%%%%%%%%%%%%%%%%%%%%%%%%%%%%%%%%%%%%%%%%%%%%%%%%%%%%%%%%%%%%%%%

Next, we consider the second-order correction to the ground state energy from $g^{2} H_{\rmm{Nic}}$.   
Since the perturbation itself is proportional to $g^{2}$ and the energy denominators are $O(g)$, the correction is proportional to $g^3$ 
in contrast to the naive expectation $g^{4}$.    
As has been mentioned in Sec.~\ref{sec:1st-order-in-Hnic}, 
the ground state is given by $U_{c} \ket{0}_c$, and the nontrivial contributions from the intermediate states involving $d$-fermions are given by 
[see Eq.~\eqref{BGGtrf}]
\begin{align}
U_{c} \, \hc{c}_{p_1}\hc{c}_{p_2}\ket{0}_c \; ,
\end{align}
with $p_1$ and $p_2$ labeling momenta of the $c$-fermions.

If we introduce 
\begin{equation}
\begin{split}
&\mathcal{M}\qty(p_1,p_2;  i_1,i_2)  \\
& :=
\rmm{sgn}\qty(p_2)\eno{-\im \qty(p_1 i_1+p_2 i_2)}-
\rmm{sgn}\qty(p_1)\eno{-\im \qty(p_1 i_2+p_2 i_1)} \; , 
\end{split}
\end{equation}
the matrix element that gives the effective Hamiltonian for the $d$-fermions is given by
\begin{align}
& {}_{c}\!\bra{0}c_{p_2}c_{p_1}\hc{U}_{c} H_{\rmm{Nic}}U_{c} \ket{0}_c \nt
=&-\frac{\im}{L}\sum_{i=1}^L\Bigl[ 
\eno{-\im \qty(p_1+p_2)i}\Bigl\{
\frac{1}{\im \pi}\qty(\eno{-\im p_1}-\eno{-\im p_2})\theta\qty(p_1p_2)\nt
&+\frac{1}{4}
\qty(\rmm{sgn}\qty(p_1)-\rmm{sgn}\qty(p_2))
\qty(1-\eno{-\im \qty(p_1+p_2)})
\Bigl\}  \nt
&+ \mathcal{M}\qty(p_1,p_2;  i,i+2)\hc{d}_{i+1}d_i  
+\mathcal{M}\qty(p_1,p_{2}; i+2,i)\hc{d}_{i}d_{i+1} \nt
&+\left\{ \mathcal{M}\qty(p_1,p_2; i,i)+\mathcal{M}\qty(p_1,p_2; i+1,i+1)-1 \right\} n^d_i 
\Bigr] \; .
\end{align}

By using this matrix element, the second-order perturbation energy including $d$-fermions is given by 
\begin{equation}
E_{\text{g.s.}}^{\qty(2)} = 
-\frac{g^{3}}{2} \sum_{p_1,p_2}\frac{\abs{_{c}\!\bra{0}c_{p_2}c_{p_1}\hc{U}_{c} H_{\rmm{Nic}}U_{c} \ket{0}_c
}^2}{ \abs{\sin(p_1)}+\abs{\sin(p_2)} } \; . 
\label{eqn:2nd-eff-Ham-d}
\end{equation}
Therefore, the second-order processes introduce hopping terms $\hc{d}_{i+1}d_i$ and $\hc{d}_{i}d_{i+1}$ as well as other 
interaction terms like $n_i^d n_j^d$ and $\hc{d}_{i+1} n_j^d d_i$, which generically endow the $d$-band with a finite dispersion.   

\section{Conclusion}
\label{sec:conclusion}
%%%%%%%%%%%%%%%%%%%%%%%%%%%%%%%%%%%%%%%%%%%%%%%%%%%%%%%%%%%%%%%%%%%%%%%%%
In this paper, we have investigated the ground states and low-lying excitations of a supersymmetric model of interacting lattice fermions 
paying particular attention to whether supersymmetry is broken or not in the ground state.  
The model contains a fermion-pairing term as well as the Nicolai model respectively as the first and second-order terms in a control parameter $g$.  
Remarkably, the superconducting order parameter (i.e., the pairing amplitude) is directly related to one of the order parameters ($G$) of 
SUSY breaking.   

We have shown that, for both finite and infinite system sizes, 
supersymmetry is always broken except at $g = \infty$ where the model reduces to the Nicolai model.  
To prove this, we have exploited the existence of an operator $\chi$ which has a non-trivial anti-commutator with the supercharge $Q_{\text{S}}(g)$. 
We also physically interpreted this by numerically calculating the superconducting [$\Delta(g)$] and SUSY-breaking ($G$) order parameters.   

In relativistic systems, it is known that spontaneously broken SUSY implies a fermion with gapless $p$-linear dispersion (the NG fermion).  
To see if this is the case in our lattice fermion model, we have applied the single-mode approximation 
and confirmed that the spectrum of the NG fermion is bounded from above by a dispersion linear 
in momentum $p$, implying the existence of a lattice counterpart of the NG fermion in our model.  

With numerical diagonalization and perturbative calculations, we have investigated the extensive ground-state degeneracy 
(which is exponentially large in the system size) and the zero-energy flat band that is suggested for sufficiently small $g$.  
Our numerical results indicate that the zero-energy flat band and the associated extensive degeneracy are not robust 
against finite $g$; the extensive degeneracy observed at $g=0$ and $\infty$ is resolved almost completely leaving $O(1)$ degeneracy 
associated with the existence of fermion zero modes.  We then tried to understand how the interactions 
between the $c$ and $d$ fermions affects the dispersion of the $d$-fermion 
using perturbative argument in $g$, and showed that processes second order in the interaction part $g^{2} H_{\rmm{Nic}}$ turn 
the flat $d$-fermion band dispersive.   

Although our study has focused on the mathematical aspects of a supersymmetric lattice model, 
future work could investigate the physical implications of SUSY breaking. Specifically, it would be interesting 
to explore the effects of supersymmetry breaking on the transport properties of  supersymmetric systems, 
such as the electrical and thermal conductivities.  Also, a systematic study focusing more on general aspects (dispersion, etc.) 
of the non-relativistic Nambu-Goldstone fermions would be highly desirable, as they are crucial in understanding the mechanism 
of supersymmetry breaking. 

%Finally, it would be worthwhile to examine the effects of fluctuations beyond mean-field theory, such as quantum and thermal fluctuations, which may have important implications for the system's behavior near the SUSY-breaking transition.

%%%%%%%%%%%%%%%%%%%%%%%%%%%%%%%%%%%%%%%%%%%%%%%%%%%%%%%%%%%%%%%%%%%%%%%%%
\section*{ACKNOWLEDGEMENTS}
%%%%%%%%%%%%%%%%%%%%%%%%%%%%%%%%%%%%%%%%%%%%%%%%%%%%%%%%%%%%%%%%%%%%%%%%%
%I would like to express my sincere gratitude to Associate Professor Keisuke Totsuka for his patient and enthusiastic guidance and discussions during this study.
The author (UM) is sincerely grateful to Kenji Simomura, who provided valuable comments on several occasions. 
The authors also thank Hosho Katsura and Yu Nakayama for helpful discussions and correspondence.  
The author (KT) is supported in part by Japan Society for the Promotion of Science (JSPS) KAKENHI Grant No. 21K03401. 

%%%%%%%%%%%%%%%%%%%%%%%%%%%%%%%%%%%%%%%%%%%%%%%%%%%%%%%%%%%%%%%%%%%%%%%%%
\appendix
%%%%%%%%%%%%%%%%%%%%%%%%%%%%%%%%%%%%%%%%%%%%%%%%%%%%%%%%%%%%%%%%%%%%%%%%%
%%%%%%%%%%%%%%%%%%%%%%%%%%%%%%%%%%%%%%%%%%%%%%%%%%%%%%%%%%%%%%%%%%%%%%%%%
\section{Diagonalization of $H_{\rmm{pair}}$ and its ground-state energy}
\label{App:freetaikakuka}
%%%%%%%%%%%%%%%%%%%%%%%%%%%%%%%%%%%%%%%%%%%%%%%%%%%%%%%%%%%%%%%%%%%%%%%%%
In the small $g$ limit, the first non-trivial contribution comes from the pairing part $g H_{\rmm{pair}}$.  
In this appendix, we diagonalize the pairing part $g H_{\rmm{pair}}$ of the Hamiltonian.  
By applying the discrete Fourier transform $c_i=\frac{1}{\sqrt{L}}\sum_{k} \eno{\im k i}c_k$, 
we can rewrite $H_{\rmm{pair}}$ in the Bogoliubov-de Genues (BdG) form: 
\begin{equation}
\label{bdg-diag}
H_{\rmm{pair}}
=\sum_{k>0} 
    \qty(\hcc_{k},c_{-k})
    \begin{pmatrix}
    0 & 2\im \sin(k) \\
    -2\im \sin(k) & 0
    \end{pmatrix}
    \begin{pmatrix}
    c_{k} \\
    \hcc_{-k}
    \end{pmatrix}  \; . 
\end{equation}
If we define the Bogoliubov quasiparticle $f_{k}$ as:
\begin{equation}
\label{Bogoliubov}
%    \begin{pmatrix}
%    c_{k} \\
%    \hc{c}_{-k}
%    \end{pmatrix}
%    :=\frac{1}{\sqrt{2}}
%    \begin{pmatrix}
%    1 & -\im \\
%    -\im & 1
%    \end{pmatrix}
%    \begin{pmatrix}
%    f_{k} \\
%    \hc{f}_{-k}
%    \end{pmatrix} 
%     \notag \\
    \begin{pmatrix}
    f_{k} \\
    \hc{f}_{-k}
    \end{pmatrix}
    =\frac{1}{\sqrt{2}}
    \begin{pmatrix}
    1 & \im \\
    \im & 1 
    \end{pmatrix}
    \begin{pmatrix}
    c_{k} \\
    \hc{c}_{-k}
    \end{pmatrix}
    \; ,
\end{equation}
the pairing part $H_{\rmm{pair}}$ takes the following diagonal form:
\begin{equation}
\begin{split}
& H_{\rmm{pair}} \\
& = \sum_{k>0} 
    \qty(\hc{f}_{k},f_{-k})
    \begin{pmatrix}
    2\sin(k) & 0 \\
    0 & -2\sin(k)
    \end{pmatrix}
    \begin{pmatrix}
    f_{k} \\
    \hc{f_{-k}}
    \end{pmatrix} \\
% & =\sum_{k>0} 2\sin(k) \qty(
%    \hc{f}_{k} f_{k}+ \hc{f}_{-k} f_{-k}-1) \; .
& =\sum_{k} 2 |\sin(k)| \qty(
    \hc{f}_{k} f_{k} - 1/2) 
 \end{split}
 \label{eqn:Hfree-diagonalized}
\end{equation}
(in passing to the last line, we have extended the $k$-summation to $(-\pi,\pi]$).   
%Hence, the single-particle energy spectrum $\eps_{\rmm{pair}}(k)$ is given by
%\begin{equation}
%\eps_{\rmm{free}}(k)=2g | \sin(k) |
%\quad  \qty(-\pi<k \leq \pi) \; . 
%\end{equation}

The ground state of $g H_{\rmm{pair}}$ is the vacuum of the Bogoliubov quasiparticle $f$, its energy is given by  
the zero-point energy of \eqref{eqn:Hfree-diagonalized}:
\begin{equation}
e^{\rmm{pair}}_{0} 
= - \int^{\pi}_{-\pi}\frac{\rmm{d}k}{2\pi} g | \sin(k) | 
=-\frac{2g}{\pi} \; .
\label{eqn:Egs-density}
\end{equation}

For later analyses, we next rewrite the Bogoliubov transformation \eqref{Bogoliubov} and the ground state of 
$H_{\rmm{pair}}$ in terms of unitary transformation.  Specifically, if we define
\begin{align}
\label{Ubogo}
    \left\{
      \begin{aligned}
      & Q_{k}:=\theta_{k} \qty(c_{k}c_{-k}+\hc{c}_{-k}\hc{c}_{k}) \\
      & U_{c} :=\prod_{k>0} \exp(\im \theta_{k} Q_{k}) \; ,
      \end{aligned}
  \right.
\end{align}
the Bogoliubov quasiparticles $f_{k}$ are given for $k>0$ as:
\begin{align}
\label{BGGtrf}
\left\{
\begin{aligned}
& f_k=U_{c} c_{k}\hc{U}_{c}=\cos(\theta_{k})c_{k}+\im \sin(\theta_{k}) \hc{c_{-k}} \\
& \hc{f}_{-k}=U_{c} \hc{c_{-k}} \hc{U}_{c} =\im \sin(\theta_{k}) c_{k}+\cos(\theta_{k})\hc{c_{-k}}  \; .
\end{aligned}
\right.
\end{align}
Therefore, if we take
\begin{align}
\theta_{k}=\frac{\pi}{4} \; ,
\end{align}
Eq.~\eqref{BGGtrf} reproduces \eqref{Bogoliubov}.  
 It then follows that the ground state of $H_{\rmm{pair}}$ (i.e., the vacuum of the $f$-fermions $\ket{0}_{f}$) 
 is given by $\ket{0}_{f} = U_{c} \ket{0}_{c}$ with $\ket{0}_{c}$ being the vacuum of $c_{k}$:
\[
\begin{split}
& \hc{U}_{c} H_{\rmm{pair}} U_{c} \ket{0}_{c} =e_{\text{g.s.}}^{\rmm{pair}}\ket{0}_{c} \\
& \naraba 
H_{\rmm{pair}} U_{c} \ket{0}_{c} =e_{\text{g.s.}}^{\rmm{pair}}U_{c} \ket{0}_{c}  \;  .
\end{split}
\]

%%%%%%%%%%%%%%%%%%%%%%%%%%%%%%%%%%%%%%%%%%%%%%%%%%%%%%
\section{Variational argument for $g_\rmm{c}$}
\label{sec:variational}
%%%%%%%%%%%%%%%%%%%%%%%%%%%%%%%%%%%%%%%%%%%%%%%%%%%%%%  
In general, if the Hamiltonian can be decomposed into several parts (which are not commuting in general): 
$H=\sum_i H_i$, the ground state $\ket{0}$ of the total Hamiltonian $H$ generally includes excited states of each $H_i$.    
Then, the following inequality for the total ground-state energy $e_{\text{g.s.}}$ and those [$E_{\text{g.s.}}^\qty(i)$] of the partial Hamiltonians $H_i$ 
follows from the variational principle: 
\begin{equation}
\label{g.s.ineq}
E_{\text{g.s.}}=\bra{0} H \ket{0} 
=\sum_i \bra{0} H_i \ket{0} 
\geq \sum_i E_{\text{g.s.}}^\qty(i)
\end{equation}
(the equality holds when the ground state $\ket{0}$ optimizes all $H_i$ simultaneously). 
That is, the ground-state energy of the full Hamiltonian is greater than or equal to the sum of the ground-state energies $E_{\text{g.s.}}^\qty(i)$
of the partial Hamiltonians.

Now we apply the above argument to the Hamiltonian \eqref{BCS-like}.  
Using the fact that the ground-state energy density of $g^2H_{\rmm{Nic}}$ and $gH_{\rmm{pair}}$ are given respectively by $0$
\cite{moriya2018ergodicity,Katsura_2020} and $-\frac{2g}{\pi}$ 
[see Eq.~\eqref{eqn:Egs-density}], we obtain the following variational bound:
\begin{equation}
e_{\text{g.s.}}(g)  \geq   \frac{2}{\pi}\qty(\frac{\pi}{2}-g) \; .
\label{eqn:var-lower-bound}
\end{equation}
Thus, we see that SUSY SSB occurs at least in the interval
\begin{equation}
0 \leq g < \frac{\pi}{2}  \; .
\label{eqn:SSB-phase-variational}
\end{equation}
On the other hand, when $g \to \infty$, the model reduces to the Nicolai model,  
where unbroken SUSY is established rigorously \cite{moriya2018ergodicity,Katsura_2020} 
(see Appendix~\ref{App:groundstatecount} for a quick summary). 
Therefore, assuming that there is a single phase transition, we can conclude that the critical point $g_{\text{c}}$ 
above which the SUSY-symmetric phase persists satisfies $\pi/2 \leq g_{\text{c}} \leq \infty$.

%%%%%%%%%%%%%%%%%%%%%%%%%%%%%%%%%%%%%%%%%%%%%%%%%%%%%%%%%%%%%%%%%%%%%%%%%
\section{Ground state at $g=\infty$ (Nicolai Model)}
\label{App:groundstatecount}
%%%%%%%%%%%%%%%%%%%%%%%%%%%%%%%%%%%%%%%%%%%%%%%%%%%%%%%%%%%%%%%%%%%%%%%%%

At $g=\infty$, the Hamiltonian $H_{\rmm{S}} \qty(g)/g^2$ reduces to  $H_{\rmm{Nic}}$ since $L$ and $g H_{\rm pair}$ can be neglected. Furthermore, the supercharge operator $Q_{\rmm{S}} \qty(g) /g$ now reads   
$Q_{\rmm{Nic}}= \sum_{i} \hc{c}_{i}\hc{c}_{i+1}d_i$, which, with due redefinition of the fermions $\hc{c}_{i} \migi c_{2i-1}$, $d_{i} \migi - \hc{c}_{2i}$, 
reduces to the supercharge of the Nicolai model \cite{nicolai1976supersymmetry}.   
Then, the Hamiltonian $g^2H_{\rmm{Nic}}$ is expressed as the anti-commutator of $Q_{\rmm{Nic}}$ and $\hc{Q}_{\rmm{Nic}}$:
\begin{align}
H_{\rmm{Nic}} = \acom{Q_{\rmm{Nic}}}{\hc{Q}_{\rmm{Nic}}}.
\end{align} 
It is known \cite{moriya2018ergodicity,Katsura_2020} that in the Nicolai model, spontaneous breaking of supersymmetry does not occur, 
and hence the energy of the ground state is exactly zero
(the vacuum of the $c$ and $d$ fermions is a trivial zero-energy state that saturates the lower bound of the energy \footnote{%
However, this zero-energy state is {\em not} the unique ground state.  As has been pointed out 
in Refs.~\cite{moriya2018ergodicity,Katsura_2020}, the ground state of the Nicolai model is extensively degenerate.}).  
This fact is used in the variational argument in Sec.~\ref{sec:variational}.

%\modified{
%%%%%%%%%%%%%%%%%%%%%%%%%%%%%%%%%%%%%%%%%%%%%%%%%%%%%%%%%%%%%%%%%%%%%%%%%
\section{Exact eigenstates}
\label{App:EE}
%%%%%%%%%%%%%%%%%%%%%%%%%%%%%%%%%%%%%%%%%%%%%%%%%%%%%%%%%%%%%%%%%%%%%%%%%
In our model \eqref{BCS-like}, exact eigenstates with eigenvalue $L$ exist under periodic boundary conditions when $L$ is even, as well as under open boundary conditions for arbitrary $L$. These eigenstates are simultaneous eigenstates of both $H_{\text{pair}}$ and $H_{\text{Nic}}$ with eigenvalue 0.

Recalling the expressions:
\begin{align}
&H_{\text{Nic}}=\acom{Q_{\text{Nic}}}{\hc{Q}_{\text{Nic}}} \nonumber \\
&Q_{\text{Nic}}=\sum_{j=1}^L \hc{c}_{i}\hc{c}_{i+1}d_i \nonumber \\
&\hc{Q}_{\text{Nic}}=-\sum_{j=1}^L c_{i}c_{i+1}\hc{d}_i
\end{align}
in the $c$-fermion sector, it can be verified that the ``charge-density wave'' states where the fermion numbers alternate between 1 and 0 in real space:
\begin{equation}
\label{eqn:staggered-c-sector}
\begin{split}
&\ket{0,1,0,1,0,\dots}_c \nonumber \\
&\ket{1,0,1,0,1,\dots}_c
\end{split}
\end{equation}
are annihilated by the action of both $Q_{\text{Nic}}$ and $\hc{Q}_{\text{Nic}}$. Hence, the states in \eqref{eqn:staggered-c-sector} are the 0-energy ground states of $H_{\text{Nic}}$.

As both $\hc{c}_{i}\hc{c}_{i+1}$ and $c_{i+1} c_i$ annihilate $|10 \rangle$ and $|01\rangle$, it is easy to see that the states in \eqref{eqn:staggered-c-sector} are also 0-energy eigenstates of $H_{\text{pair}}$. 
Therefore, for any choice of the $d$-fermion state $\ket{\psi_n}_d$, the states of the form
\begin{equation}
\label{eqn:staggered-c,d-sector}
\begin{split}
\ket{\Psi}_{1}:= &\ket{0,1,0,1,0,\dots}_c \otimes \ket{\psi_n}_d  \\
\ket{\Psi}_{2}:= &\ket{1,0,1,0,1,\dots}_c \otimes \ket{\psi_n}_d
\end{split}
\end{equation}
are eigenstates of $H_{\text{S}}(g)$ with eigenvalue $L$ \footnote{%
The condition that $L$ must be even is necessary for the alternating configurations \eqref{eqn:staggered-c-sector} 
to be compatible with the boundary condition.}:
\begin{align}
H_{\text{S}}(g)\ket{\Psi}_{1,2}&=\qty(L+gH_{\text{pair}}+g^2H_{\text{Nic}})\ket{\Psi}_{1,2} \nonumber \\
&=L \ket{\Psi}_{1,2}  \; .
\end{align}
As there are $2^{L}$ of the $d$-fermion states $\ket{\psi_n}_d$, the eigenvalue $L$ is at least 
$2\times2^L=2^{L+1}$-fold degenerate.   According to the numerical results from exact diagonalization, 
there seem to be more states with the energy $L$. 
%}

%%%%%%%%%%%%%%%%%%%%%%%%%%%%%%%%%%%%%%%%%%%%%%%%%%%%%%%%%%%%%%%%%%%%%%%%%
\section{Mapping to a spin-1/2 chain}
\label{App:JW}
%%%%%%%%%%%%%%%%%%%%%%%%%%%%%%%%%%%%%%%%%%%%%%%%%%%%%%%%%%%%%%%%%%%%%%%%%
The model $H_{\rmm{S}} \qty(g)$ in Eq.~\eqref{BCS-like} describes a system of one-dimensional lattice spinless fermions that contains 
various interactions.  In this Appendix, we show that, when restricted to particular sectors, the model can be mapped 
to a tractable spin-1/2 model via the Jordan-Wigner (JW) transformation.    
The boundary effects are ignored for simplicity.  We first introduce two species of spins 
$\sigma_i^{-}$ and $\tau_i^{-}$ corresponding to $c_i$ and $d_i$, respectively. 
Then, the JW transformation is defined as 
\begin{align}
    \left\{
      \begin{aligned}
      & c_i:=\qty(\sz_1 \tz_1) \cdots \qty(\sz_{i-1} \tz_{i-1})\sigma^{-}_i \\
      & d_i:=\qty(\sz_1 \tz_1) \cdots \qty(\sz_{i-1} \tz_{i-1}) \qty(-\sz_i) \tau^{-}_{i}  \; ,
      \end{aligned}
  \right.
\end{align}
where $\sz_i$ and $\tz_i$ are related to the fermion numbers as
\begin{align}
    \left\{
      \begin{aligned}
      & n^c_i:=\hf \qty(1+\sz_i) \\
      & n^d_i:=\hf \qty(1+\tz_i)  \; .
      \end{aligned}
  \right.
\end{align}
With these equations, the Hamiltonian is now recast as:
\begin{equation}
\begin{split}
   H_{\rmm{S}} \qty(g) = &
    -\frac{g}{2} \sum_{i=1}^L \tz_i \qty(
   \sx_i \sx_{i+1}-\sy_i \sy_{i+1}
    ) \\
    & - g^2 \sum_{i=1}^L \qty(
    \sigma^{+}_{i} \tau^{-}_i \tau^{+}_{i+1} \sigma^{-}_{i+2} + \rmm{H.c.}
    )  \\
    & + \frac{g^2}{4}\sum_{i=1}^L\qty{
    \sz_{i}\sz_{i+1}-\tz_{i}\qty(\sz_{i}+\sz_{i+1})
    }  \\
    & +\qty(1+\frac{g^2}{4})L \; ,
\end{split}
\label{eqn:equivalent-spin-Ham}
\end{equation}
where $L$ is the total number of lattice sites.

The resulting Hamiltonian contains three and four-spin interactions and is not solvable in general.  
However, if we note that the model \eqref{eqn:equivalent-spin-Ham} conserves 
$\sum_{i}\tau^{z}_{i}$ ($\sum_{i} n^{d}_{i}$ in the fermion language), 
we see that within the two special sectors with $\sum_{i}\tau^{z}_{i}=  L$ and $- L$ (fully polarized $\tau$),  
the spin-1/2 model \eqref{eqn:equivalent-spin-Ham} reduces to a more tractable one.  
Specifically, within the sector with all the $\tau$-sites occupied by $\tz_i= 1$ or $-1$ (corresponding to $n_i^d=1$ or $0$, respectively), 
the four-spin part vanishes and the Hamiltonian $H_{\rmm{S}}^{\pm}(g)$ reduces to:
\begin{equation}
\begin{split}
H_{\rmm{S}}^{\pm}(g) = & \mp \frac{g}{2} \sum_{i}\qty(
\sigma_i^x\sigma_{i+1}^x-\sy_i\sy_{i+1}
)   + \frac{g^2}{4}\sum_{i} \sigma_{i}^z \sigma_{i+1}^z  \\
& \mp \frac{g^2}{2}\sum_{i} \sigma_{i}^z  
+\qty(1+\frac{g^2}{4})L \; .
\end{split}
\end{equation}

By applying the unitary transformation:
\begin{align}
U:=\sigma_{1}^x\sigma_{3}^x\sigma_{5}^x\cdots
\end{align}
that flips the signs of $\sigma_{2i-1}^y$ and $\sigma_{2i-1}^z$, 
we obtain the transformed Hamiltonian $\widetilde{H}^{\pm}_{\rmm{S}}(g)$ 
which is nothing but the well-known $S=1/2$ XXZ model in a staggered magnetic field:
%\begin{equation}
%\begin{split}
%\widetilde{H}^{\pm}_{\rmm{S}}(g) : = & \hc{U}H_{\pm} U  \\
%= & \mp \frac{g}{2} \sum_{i=1}^L\qty(
%\sigma_i^x\sigma_{i+1}^x+\sigma_i^y\sigma_{i+1}^y
%)   \\
%& +  \frac{g^2}{4}\sum_{i=1}^L\qty{
%\sigma_{i}^z\sigma_{i+1}^z\pm \qty(\sigma_{2i-1}^z-\sigma_{2i}^z)
%}  \\
%& +  \qty(1+\frac{g^2}{4})L \; .
%\end{split}
%\end{equation}
%%%

\begin{equation}
\begin{split}
\widetilde{H}_{\rmm{S}}^{\pm}(g) := & \hc{U}H_{\rmm{S}}^{\pm}(g) U  \\
= & \mp \frac{g}{2} \sum_{i}\qty(
\sigma_i^x\sigma_{i+1}^x + \sy_i\sy_{i+1}
)   
- \frac{g^2}{4}\sum_{i} \sigma_{i}^z \sigma_{i+1}^z  \\
& \mp \frac{g^2}{2}\sum_{i} (-1)^{i}  \sigma_{i}^z  
+\qty(1+\frac{g^2}{4})L \; .
\end{split}
\label{eqn:XXZ-staggered-F}
\end{equation}  
The sign of the first term is irrelevant as we can always make it positive by a staggered $\pi$-rotation $\prod_{i} (\im \sigma^{z}_{2i-1})$ 
along the $z$-axis.

%%%%%%%%%%%%%%%%%%%%%%%%%%%%%%%%%%%%%%%%%%%%%%%%%%%%%%
\section{Superfield formalism}
\label{App:chobakeisiki}
%%%%%%%%%%%%%%%%%%%%%%%%%%%%%%%%%%%%%%%%%%%%%%%%%%%%%%
The superfield formalism is a method for describing supersymmetric multiplets as superfields on a super space $\qty(x^{\mu},\theta_a,\Bar{\theta}_a)$, which includes virtual Grassmann coordinates $\qty(\theta_{a},\Bar{\theta}_{a})$ in addition to the spacetime coordinates $x^{\mu}=\qty(t,x,\cdots)$
\cite{wess1992supersymmetry,SUSYfracton,weinberg1995quantum}. 
When the supercharges $Q$ consist of one or three fermions [this is the case in the model \eqref{eqn:def-super-charge}], 
the corresponding models can be described using chiral superfields, which have fermionic statistics.   
This allows for a simple expression of supersymmetric transformations.
%%%%%%%%%%%%%%%%%%%%%%%%%%%%%%%%%%%%%%%%%%%%%%%%%%%%%%%
\subsection{Superfield formalism: general construction}
\subsubsection{Supersymmetric Transformations and Super-covariant Derivatives}
%%%%%%%%%%%%%%%%%%%%%%%%%%%%%%%%%%%%%%%%%%%%%%%%%%%%%%%
Consider a function of time $t$ and a pair of complex-conjugate Grassmann numbers $\theta$ and $\Bar{\theta}$: 
a superfield $\Phi(t,\theta,\Bar{\theta})$. The relations between $t$, $\theta$, and $\bar{\theta}$ are given by
\begin{align}
\left\{
\begin{aligned}
& \com{t}{\theta}=\comm{t}{\Bar{\theta}}=0 \\
& \acom{\theta}{\Bar{\theta}}=0 \; , \quad \theta^2=\Bar{\theta}^2=0 \; .
\end{aligned}
\right.
\end{align}
By the nilpotency, any functions of Grassmann variables have at most linear dependence on them.  
The Grassmann differential satisfies anti-commutation relations that are analogous to the canonical commutation relation 
$\comm{\partial}{x}=1$, namely
\begin{align}
\left\{
\begin{aligned}
&\acom{\theta}{\del{}{\theta}}
=\acom{\bar{\theta}}{\del{}{\bar{\theta}}}
=1 \\
& \acom{\theta}{\del{}{\bar{\theta}}}
=\acom{\bar{\theta}}{\del{}{{\theta}}}=0
\end{aligned}
\right.
\end{align}
All other commutation relations are defined to be commutative. Furthermore, integration to Grassmann variables is defined to be the same as differentiation.

As differential operators on the superfield, we define the supersymmetric transformation $\mathcal{Q}$ and the super-covariant derivative $\mathcal{D}$ as follows:
\begin{align}
\left\{
\begin{aligned}
& \mathcal{Q}:=\del{}{\theta}+\im \Bar{\theta} \del{}{t} \; , \quad 
\Bar{\mathcal{Q}}:=\del{}{\Bar{\theta}}+\im \theta \del{}{t} \\
& \mathcal{D}:=\del{}{\theta}-\im \Bar{\theta} \del{}{t} \; , \quad 
\Bar{\mathcal{D}}:=\del{}{\Bar{\theta}}-\im \theta \del{}{t}
\end{aligned}
\right.
\end{align}
Then, we have the following relations:
\begin{align}
\left\{
\begin{aligned}
& \mathcal{Q}^2=\Bar{\mathcal{Q}}^2=\mathcal{D}^2=\Bar{\mathcal{D}}^2=0 \\
& \acom{\mathcal{Q}}{\Bar{\mathcal{Q}}}=2\im \del{}{t} \; , \quad  
\acom{\mathcal{D}}{\Bar{\mathcal{D}}}=-2\im \del{}{t} \\
& \acom{\mathcal{Q}}{\mathcal{D}}=\acom{\mathcal{Q}}{\Bar{\mathcal{D}}}
= \acom{\Bar{\mathcal{Q}}}{\mathcal{D}}=\acom{\Bar{\mathcal{Q}}}{\bar{\mathcal{D}}}=0
\end{aligned}
\right.
\end{align}

In this way, it is found that $\acom{\mathcal{Q}}{\Bar{\mathcal{Q}}}$ is consistent with $\acom{Q}{\hc{Q}}=H$ to generate time translations. Also, the super-covariant derivative behaves covariantly, as its name suggests, under the super-symmetric transformation.

%%%%%%%%%%%%%%%%%%%%%%%%%%%%%%%%%%%%%%%%%%%%%%%%%%%%%%%%%%
\subsection{Application to the model $H_{\rmm{S}} \qty(g)$}
\label{App:chobakeisiki2}
%%%%%%%%%%%%%%%%%%%%%%%%%%%%%%%%%%%%%%%%%%%%%%%%%%%%%%%%%%
The spinless fermion model \eqref{BCS-like} we consider can be described by superfields 
that satisfy both Fermi statistics and chiral conditions. Let $C_i=C_i(t,\theta,\bar{\theta})$ be a superfield that satisfies Fermi statistics. The chiral condition requires the conjugate super-covariant derivative $\bar{\mathcal{D}}$ to annihilate $C_i$:
\begin{align}
\label{chiralC}
0:=\bar{\mathcal{D}}C_{i}=\qty(\del{}{\Bar{\theta}}-\im \theta \del{}{t})C_{i}  \; .
\end{align}
Here, we have $\bar{\mathcal{D}}\theta=0$ and $\bar{\mathcal{D}}\qty(t+\im\bar{\theta}\theta)=0$. Therefore, by allowing $C_i$ to depend only on $\im\bar{\theta}\theta$ and $\theta$, the condition \eqref{chiralC} is automatically satisfied. 
Expanding $C_i$ in powers of $\im\bar{\theta}\theta$ and $\theta$, we obtain the component representation:
\begin{align}
C_{i}=c_{i}\qty(t) +\theta \sqrt{2} F_{i}\qty(t) + \bar{\theta} \theta \im \Dot{c}_{i}\qty(t)  \; .
\end{align}

This representation can also be obtained by expanding $C_i$ as $C_i\qty(t,\theta,\bar{\theta})$ and determining the component fields to make it vanish under $\bar{\mathcal{D}}$. By identifying $c_i$ as the annihilation operator of the fermion, we obtain the action corresponding to the spinless fermion model.

Furthermore, we define $\hc{C}$ as:

\begin{align}
\hc{C}_{i}=\hc{c}_{i}+ \bar{\theta}\sqrt{2}\hc{F}_{i}-\bar{\theta}\theta \im \hc{\Dot{c}_{i}}
\end{align}

This is also a superfield, but not a chiral superfield. We also define the charge conjugate transformation of $C$ as the operation of taking the Hermitian conjugate of the operators:

\begin{align}
C_{i}^{c}:=\hc{c}_{i}+ \theta\sqrt{2} \hc{F}_{i}+\bar{\theta}\theta \im \hc{\Dot{c}_{i}}
\end{align}

This is also a chiral superfield.

By computing the supersymmetric transformation of $C_i$, we can confirm the construction of supersymmetric invariant actions. Specifically, we have
\begin{align}
\label{SUSYtrfC}
\left\{
\begin{aligned}
& \mathcal{Q} C_{i}
=\sqrt{2}F_{i}-\bar{\theta} \im \Dot{c}_{i}+\im \bar{\theta} \Dot{c}_{i}=\sqrt{2}F_{i} \\
& \bar{\mathcal{Q}}C_{i}=\theta \im \Dot{c}_{i}+\im \theta \Dot{c}_{i} =\theta 2\im \Dot{c}_{i} \\
\end{aligned}
\right.
\end{align}
which shows that the $\theta$ term of chiral superfields is invariant under supersymmetric transformations that are at most times derivatives. Therefore, in constructing supersymmetric invariant actions, we can use not only the $\bar{\theta}\theta$ term but also the $\theta$ term of chiral superfields. Furthermore, by taking the Hermitian conjugate of each component, we can see that while $\hc{C}_{i}$ is not a chiral superfield, the $\bar{\theta}\theta$ term also undergoes supersymmetric transformations that are at most time derivatives. Hence, if we keep in mind that $C_{i}$ carries fermionic statistics and $C_{i}^2$ should vanish, we can construct the Lagrangian as the sum of the ($\bar{\theta}\theta$ integral of $\hc{C}C$), the (odd order $\theta$ integral of $C_i$'s), and their Hermitian conjugates.

 Specifically, we can construct the fermion kinetic term by adding the following term:

\begin{align}
\sum_{i=1}^L \intd \theta \rmm{d}\bar{\theta} \hf
\hc{C}_{i}C_{i}=\sum_{i} \qty(\im \hc{c}_{i}\Dot{c}_{i}+\hc{F}_{i}F_{i}+\qty(\rmm{t.d.}) )
\end{align}

In our model \eqref{BCS-like}, we add the fermionic chiral superfield $D_{i}$ in addition to $C_{i}$:

\begin{align}
D_{i}=d_{i}+\theta \sqrt{2} G_{i}+\bar{\theta}\theta \im \Dot{d}_{i}
\end{align}

If we formally replace $c$ and $d$ fermions in the supercharge \eqref{eqn:def-super-charge} with the corresponding 
chiral superfields $C$ and $D$ in the above, we can write the following term in the Lagrangian:  
\begin{align}
\label{chobatheta}
& \sum_{i=1}^L \intd \theta \qty(1-g C_{i}C_{i+1} ) D_{i}^{c} + \rmm{H.c.} \notag\\
&=\sqrt{2} \sum_{i=1}^L \qty(
\hc{G}_{i} \qty(1-gc_{i} c_{i+1} )+
F_{i} \qty(
\hc{\hc{d_{i}}c_{i+1}-d_{i-1}}c_{i-1}
)
)\nt
&\phantom{=}     + \rmm{H.c.}
\end{align}

This term can be expressed using superfields and corresponds to the Hamiltonian in the action. 
The action corresponding to the Hamiltonian can be given as:
\begin{equation}
\begin{split}
S:=& \sum_{i=1}^L \intd t \intd \theta \rmm{d}\bar{\theta}
\hf \qty(\hc{C}_{i}C_{i}+\hc{D}_{i}D_{i})  \\
& - \frac{1}{\sqrt{2}} \sum_{i} \intd t \intd \theta
\qty(1-g C_{i}C_{i+1} ) D_{i}^{c} + \rmm{H.c.}  \\
=& \intd t \sum_{i=1}^L \qty(
\im \hc{c}_{i}\Dot{c}_{i}+\im \hc{d_{i}}\Dot{d}_{i}
+\abs{F_{i}}^{2}+\abs{G_{i}}^{2}9
)  \\
& -\intd t \sum_{i=1}^L g F_{i} \qty(
\hc{d_{i}} c_{i+1}-\hc{d}_{i-1}c_{i-1}
)+ \rmm{H.c.} \\
& -\intd t \sum_{i=1}^L G_{i} \qty(1+g \hc{c}_{i}\hc{c}_{i+1} )+ \rmm{H.c.}
\end{split}
\end{equation}

This is because upon eliminating the auxiliary fields $F_{i}$ and $G_{i}$, the corresponding Hamiltonian takes the form of the Hamiltonian \eqref{BCS-like}.

Through the Legendre transformation, the Hamiltonian can be written as:
\begin{equation}
\begin{split}
H = &
\sum_{i=1}^L\qty(- \abs{F_{i}}^{2}- \abs{G_{i}}^{2})  \\
&+ \sum_{i=1}^L g F_{i} \qty(
\hc{d_{i}} c_{i+1}-\hc{d}_{i-1}c_{i-1}
)+ \rmm{H.c.}  \\
&+ \sum_{i=1}^L G_{i} \qty(1+g \hc{c}_{i}\hc{c}_{i+1} )+ \rmm{H.c.}
\end{split}
\label{formean}
\end{equation}
The following equations of motion enable us to express the auxiliary fields with $c_{i}$ and $d_{i}$:
\begin{align}
\left\{
\begin{aligned}
\label{EqMaxi}
& \hc{F}_{i}=g \qty(\hc{d_{i}} c_{i+1}-\hc{d}_{i-1}c_{i-1} ) \\
& \hc{G}_{i}=1+g \hc{c}_{i}\hc{c}_{i+1}  \;  .  \\
\end{aligned}
\right.
\end{align}
Substituting these into the Hamiltonian \eqref{formean}, we recover the original Hamiltonian \eqref{BCS-like} 
except for the constant term $g^2L/2$:
\begin{equation}
\begin{split}
H   =& \sum_{i=1}^{L} \qty(
\hf \acom{F_i}{\hc{F}_i}+\hf \acom{G_i}{\hc{G}_i}
)  \\
= & \qty(1+\frac{g^{2}}{2})L
+g\sum_{i=1}^L\qty(\hc{c}_{i}\hc{c}_{i+1}+c_{i+1}c_{i}) \\
&+g^{2} \sum_{i=1}^L\qty(
\hc{c}_{i}c_{i+2}\hc{d}_{i+1}d_{i} \HC) \\
&+g^{2} \sum_{i=1}^L\qty( -(n^{c}_{i+1}+n^{c}_{i}-1)n^{d}_{i} + n^{c}_{i}n^{c}_{i+1}
)  \;   .
\end{split}
\end{equation}

It is straightforward to apply the formalism obtained here to other supersymmetric models.  
For instance, if we use 
\begin{equation}
\label{extendedSF}
\sum_{i}\intd \theta \qty(
C_{i}-gC_{i}C_{i+1}D^c_{i}
) \HC
\end{equation}
and 
\begin{equation}
\label{Z2SF}
\sum_{i}\intd \theta \qty(
C_{i}+gC_{i-1}C_{i}C_{i+1}
) \HC  
\end{equation} 
instead of \eqref{chobatheta}, we obtain the Hamiltonian for the extended Nicolai model \cite{sannomiya2016supersymmetry} with 
a cubic fermion dispersion and the Hamiltonian for the $\mathbb{Z}_2$ Nicolai model \cite{sannomiya2017supersymmetry}, 
respectively.

\clearpage

%%%%%% REFERENCES %%%%%%%%%%%%%%%%%%%%%%%%%%%%%%%%%%%%%
%\bibliographystyle{apsrev4-2}
%\bibliography{SUSY_Ref.bib} 
%%%%%%%%%%%%%%%%%%%%%%%%%%%%%%%%%%%%%%%%%%%%%%%%%%%
% BBL-file here
\end{document}